\documentclass[letterpaper,11pt]{article}
\newif\ifNOlipics
\NOlipicsfalse
\newif\ifdraft
\draftfalse
\newif\iffull
\fullfalse
\usepackage[utf8]{inputenc}
\usepackage{microtype}
\usepackage[letterpaper,margin=1in]{geometry}
\usepackage[numbers,sort&compress]{natbib}
\usepackage{xcolor,xspace}
\usepackage{amsmath}
\usepackage{amssymb}
\usepackage{amsthm}
\usepackage{paralist}
\usepackage{enumerate}
\usepackage{textgreek}
\usepackage{graphicx}
\usepackage{framed}
\usepackage{tcolorbox}

\usepackage[unicode]{hyperref}
\usepackage[capitalize, nameinlink]{cleveref}

\usepackage{thmtools,thm-restate}

\usepackage{xspace}
\usepackage{nicefrac}
\usepackage[noend]{algorithmic}
\usepackage{algorithm}
\usepackage{dsfont}
\usepackage{colonequals}

\theoremstyle{plain}
\newtheorem{theorem}{Theorem}[section]
\newtheorem*{theorem*}{Theorem}
\newtheorem{lemma}[theorem]{Lemma}
\newtheorem*{lemma*}{Lemma}

\newtheorem{claim}[theorem]{Claim}
\newtheorem*{claim*}{Claim}

\newtheorem*{conjecture*}{Conjecture}
\newtheorem{observation}[theorem]{Observation}

\theoremstyle{definition}
\newtheorem{definition}[theorem]{Definition}

\theoremstyle{remark}

\newtheorem*{remark*}{Remark}

\usepackage{todonotes}

\newcommand{\slackgeneration}{\textsc{SlackGeneration}\xspace}

\newcommand{\CONGEST}{\ensuremath{\mathsf{CONGEST}}\xspace}
\newcommand{\LOCAL}{\ensuremath{\mathsf{LOCAL}}\xspace}

\newcommand{\eps}{\varepsilon}
\newcommand{\poly}{\operatorname{\text{{\rm poly}}}}

\newcommand{\logstar}[1]{\log^{*} #1}
\DeclareMathOperator{\E}{\mathbb{E}}

\newcommand{\dist}{\operatorname{dist}}

\newcommand{\lovasz}{Lov\'{a}sz\xspace}

\newcommand{\vbl}{\textsf{vbl}}

\newcommand{\Exp}{\mathbb{E}}

\renewcommand{\phi}{\varphi}

\newcommand{\cB}{\mathcal{B}}

\newcommand{\cD}{\mathcal{D}}
\newcommand{\cE}{\mathcal{E}}

\newcommand{\cH}{\mathcal{H}}

\newcommand{\cL}{\mathcal{L}}
\newcommand{\cN}{\mathcal{N}}

\newcommand{\cO}{\mathcal{O}}

\newcommand{\cS}{\mathcal{S}}

\newcommand{\cV}{\mathcal{V}}

\newcommand{\myparagraph}[1]{
\smallskip

\noindent\textbf{#1}}

\newcommand{\risk}{risk\xspace}

\newcommand{\black}{\ensuremath{\mathsf{black}}\xspace}
\newcommand{\white}{\ensuremath{\mathsf{white}}\xspace}
\newcommand{\assoc}{\ensuremath{\mathsf{assoc}}\xspace}

\newcommand{\true}{\ensuremath{\mathsf{true}}\xspace}
\newcommand{\false}{\ensuremath{\mathsf{false}}\xspace}

\newcommand{\Vsparse}{\ensuremath{V_{sparse}}\xspace}

\newcommand{\ouremph}[1]{\textsf{#1}}

\newenvironment{mycover}
{\list{}{\listparindent 0pt
        \itemindent    \listparindent
        \leftmargin    1cm
        \rightmargin   1cm
        \parsep        0pt}%
    \raggedright
    \item\relax}
{\endlist}

\newcommand{\myemail}[1]{\,$\cdot$\, {\small #1}}
\newcommand{\myaff}[1]{\,$\cdot$\, {\small #1}\par\smallskip}

\hypersetup{
    colorlinks=true,
    linkcolor=black,
    citecolor=black,
    filecolor=black,
    urlcolor=[HTML]{0088CC},
    pdftitle={Distributed Delta-Coloring under Bandwidth Limitations},
    pdfauthor={Yannic Maus}
}

\begin{document}

\begin{mycover}
	{\huge\bfseries\boldmath Distributed Delta-Coloring under Bandwidth Limitations \par}
	\bigskip
	\bigskip
	\bigskip
	
	\textbf{Magn\'us M. Halld\'orsson}
	\myemail{mmh@ru.is}
	\myaff{Reykjavik University, Iceland}
	
	\textbf{Yannic Maus\footnote{Supported by the Austrian Science Fund (FWF), Grant P36280-N.}}
	\myemail{yannic.maus@ist.tugraz.at}
	\myaff{TU Graz, Austria}
\end{mycover}

\thispagestyle{empty}
\begin{abstract}
We consider the problem of coloring graphs of maximum degree $\Delta$ with $\Delta$ colors in the distributed setting with limited bandwidth. Specifically, we give a $\poly\log\log n$-round randomized algorithm in the \CONGEST model. This is close to the lower bound of $\Omega(\log \log n)$ rounds from [Brandt et al., STOC '16], which holds also in the more powerful \LOCAL model. The core of our algorithm is a reduction  to several special instances of the constructive \lovasz local lemma (LLL) and the $deg+1$-list coloring problem. 
\end{abstract}
\clearpage
\thispagestyle{empty}
\tableofcontents
\clearpage

\section{Introduction}

The objective in the $c$-coloring problem is to color the vertices of a graph with $c$ such that any two adjacent vertices receive different colors. 
In the distributed setting, the $\Delta+1$-coloring problem has long been the focus of interest as the natural \emph{local} coloring problem: any partial solution can be extended to a valid full solution. It has fast $\poly(\log\log n)$-round algorithms, both in \LOCAL \cite{CLP18} and \CONGEST \cite{HKMT21}, and so does the more general \emph{deg+1}-list coloring problem (d1LC), which is what remains when a subset of the nodes has been $\Delta+1$-colored \cite{HKNT22,HNT22}.

The $\Delta$-coloring problem, on the other hand, is \emph{non-local}: fixing the colors of just two nodes can make it impossible to form a proper $\Delta$-coloring, see \Cref{fig:NiceClique} for an example.
Due to its simplicity, it has become the prototypical problem for the frontier of the unknown \cite{GHKM18,BBKO2021hideandseek}. Even the existence of such colorings is non-trivial: a celebrated result by Brooks from the '40s shows that $\Delta$-colorings exist for any connected graph that is neither an odd cycle nor a clique on $\Delta+1$ nodes \cite{brooks_1941}. 

A $\poly(\log\log n)$-round $\Delta$-coloring algorithm was recently given in \LOCAL \cite{FHM23}, but no non-trivial algorithm is known in \CONGEST. 
It is of natural interest to examine if the transition from local to non-local problems behaves differently in \LOCAL and in \CONGEST. Thus, we set out to answer the following question:

\begin{tcolorbox}
 Is there a sublogarithmic time distributed $\Delta$-coloring algorithm using small messages?
\end{tcolorbox}
In this work, we answer the question in the affirmative. We prove the following theorem.
\begin{restatable}{theorem}{thmDeltaColoring}
\label{thm:deltaColoring}
    There is a randomized $\poly\log\log n$-round \CONGEST algorithm to $\Delta$-color any graph with maximum degree $\Delta\ge 3$. The algorithm works with high probability. 
\end{restatable}
\Cref{thm:deltaColoring} nearly matches the lower bound of $\Omega(\log\log n)$ that holds in $\LOCAL$ \cite{brandt2016LLL}. In \cite{BBKO2021hideandseek}, the authors claim that in order to make progress in our understanding of distributed complexity theory, we require a $\Delta$-coloring algorithm that is genuinely different from the approaches in \cite{PS95,GHKM18}. This is due to the fact that the current state-of-the-art runtime for $\Delta$-coloring lies exactly in the regime that is poorly understood. The approaches of \cite{PS95,GHKM18} are based on brute-forcing solutions on carefully chosen subgraphs of super-constant diameter.  In contrast,  our results are based on a bandwidth-efficient deterministic reduction to a constant number of `simple' \lovasz Local Lemma (LLL) instances and $O(\log \Delta)$ instances of d1LC; the LLL is a general solution method applicable to a wide range of problems. 

  It is known that LLL is complete for sublogarithmic computation on constant-degree graphs, but its role on general graphs is widely open \cite{CP19}. Our algorithm adds to the small list of problems (see the related work section in \cite{HMP24}) that can be solved in sublogarithmic time with an LLL-type approach, even under the presence of bandwidth restrictions. 
Before continuing further, let us first detail the computational model.

In the \CONGEST model, a communication network is abstracted as an $n$-node graph of maximum degree $\Delta$, where nodes serve as computing entities and edges represent communication links. Initially, a node is unaware of the topology of the graph $G$, nodes can communicate with their neighbors in order to coordinate their actions. This communication happens in synchronous rounds where, in each round, a node can perform arbitrary local computations and send one message of $O(\log n)$ bits over each incident edge. At the end of the algorithm, each node outputs its own portion of the solution, e.g., its color in coloring problems. The \LOCAL model is identical, except without restrictions on message size.

\subsection{Technical Overview on Previous Approaches}
\label{sec:tecoverviewDelta}
Previous fast distributed $\Delta$-coloring algorithms either use huge bandwidth \cite{PS95,GHKM18} or use limited bandwidth but only work in the extreme cases of either very high-degree \cite{FHM23} or super low-degree graphs \cite{MU21}.
  Optimally, we would like to take any of these solutions and run them with minor modifications to obtain an algorithm that uses low bandwidth and works for all degrees.
This approach is entirely infeasible for the highly specialized algorithms in \cite{PS95,GHKM18,GK20}. These works crucially rely on learning the full topology of non-constant diameter subgraphs, which is impossible in  \CONGEST.

For graphs of super-low degree, i.e., at most $\poly\log\log n$, an efficient $\Delta$-coloring algorithm with low bandwidth can be deduced from the results in \cite{MU21}. In fact, the paper takes a complexity-theoretic approach and shows that any problem can be solved in sublogarithmic time with low bandwidth as long as 1) the problem is defined on low-degree graphs, 2) a given solution can be checked efficiently for correctness by a distributed algorithm, and 3) the problem admits a sublogarithmic time \LOCAL model algorithm. As such, the results are not very constructive for any specific problem like the $\Delta$-coloring problem. In fact, it is known that these generic techniques cannot be extended to problems defined on graphs with larger degrees \cite{BCMOS21}, which is the main target of our work.

Our best hope is then the $\poly\log \log n$-round \LOCAL model algorithm of \cite{FHM23}. We discuss it in detail throughout the next few pages as it motivates the design choices of our solution. Unfortunately, for maximum degrees that are at most poly-logarithmic, it relies on the prior $O(\log\Delta)+\poly\log\log n$-round \LOCAL model algorithm from \cite{GHKM18} in a black-box manner.
For large maximum degrees, however, when $\Delta$ is $\omega(\log^3 n)$, they provide a sophisticated constant-round randomized reduction to the $deg+1$-list coloring problem (d1LC) that also works with low bandwidth. 
The central ingredient in this reduction is the notion of slack. 

\subparagraph{Slack.} To reduce the $\Delta$-coloring problem to d1LC, it suffices to obtain a unit amount of \emph{slack} for each node. Namely, if two neighbors of a node are assigned the same color,  there are then more colors available to the node than its number of uncolored neighbors. 
Slack can be easily generated w.h.p.\ (for most, but not all, kinds of nodes) with a simple single-round procedure termed $\mathsf{SlackGeneration}$, as long as the graph has high degree. This observation has been used in countless papers on various coloring problems, e.g., \cite{EPS15,HSS18,CLP18,HKMT21,FHM23}. 
 For intermediate-degree graphs, this slack generation problem can be formulated as an instance of the constructive \lovasz Local Lemma (LLL), but one that seems inherently non-implementable in \CONGEST, as we explain later. 

 Recall that the LLL is a general solution method applicable to a wide range of problems. Defined over a set of independent random variables, it asks for an assignment of the variables that avoids a set of "bad" events. The original theorem \cite{LLL73} shows that such an assignment exists as long as the probability of the events to occur is sufficiently small in relation to the dependence degree of the events, i.e., the number of other events that share a variable. There is now a general \LOCAL algorithm running in $O(\log n)$ rounds of \LOCAL \cite{MT20,CPS17}, but superfast $\poly(\log\log n)$ algorithms are only known for restricted cases \cite{FGLLL17,GHK18,Davies23}. Even less is known about solvability in \CONGEST \cite{HMN22,HMP24}.

In the presented slack generation LLL, there is a bad event for each node that holds if the respective node does not obtain slack. The mentioned $\mathsf{SlackGeneration}$ works as follows. Each node gets activated with a constant probability, picks a random candidate color that it keeps if no neighbor wants to get the same color and discards otherwise (see \Cref{alg:slackgen} in \Cref{sec:deltaColoring} for details). Hence, there are random variables for each node depicting its activation status and candidate color choice. 
The main reason why this LLL cannot be directly implemented in \CONGEST is that events involve values of variables at distance 2 in the communication graph. This makes it impossible for an event node to obtain full information on the status of all its variables, an ingredient that essentially is crucial in all known sublogarithmic-time LLL algorithms. The formal meaning of the word `essential' in that sentence is extremely technical and is captured by the notion of a \emph{simulatable} LLL (see \Cref{def:simulatability}). In essence, it says that the LLL is easy enough such that event nodes can learn enough information about their variables to execute some simple primitives such as evaluating their status (does the event hold or not), resampling their variables, and computing certain conditional probabilities for the event to hold under partial variable assignments. The latter condition is the most challenging one to ensure.

\subsection{Our Technical Approach}
\label{sec:nutshell}
What we have discussed so far is only half the truth. In fact, the slack generation process only works for \emph{sparse} nodes, i.e., nodes with many non-edges in their neighborhood. If the graph is locally too dense, then slack cannot be obtained via this LLL. Thus, the algorithm of \cite{FHM23} carefully analyzes the topological structure of the hard instances for $\Delta$-coloring, combining several different (deterministic and randomized) methods to create slack.  
Such a treatment seems to be inherent to the $\Delta$-coloring problem as a very similar classification was independently and currently discovered in the streaming model \cite{AKM22}. Additionally, it has also been shown to be useful in different models of computation. In the aftermath of these works, it has been used to obtain efficient massively parallel algorithms for the problem \cite{CCDM24}.

Our algorithm is based on a fine-grained version of this classification equipped with a sequence of various LLLs for eventual slack generation. Each LLL is easier to solve in the \CONGEST model than the aforementioned slack generation LLL. In the following, we use the terminology of \cite{FHM23}, and explain their algorithm and our solution in more detail.

\begin{figure}
\centering
\includegraphics[width=\textwidth]{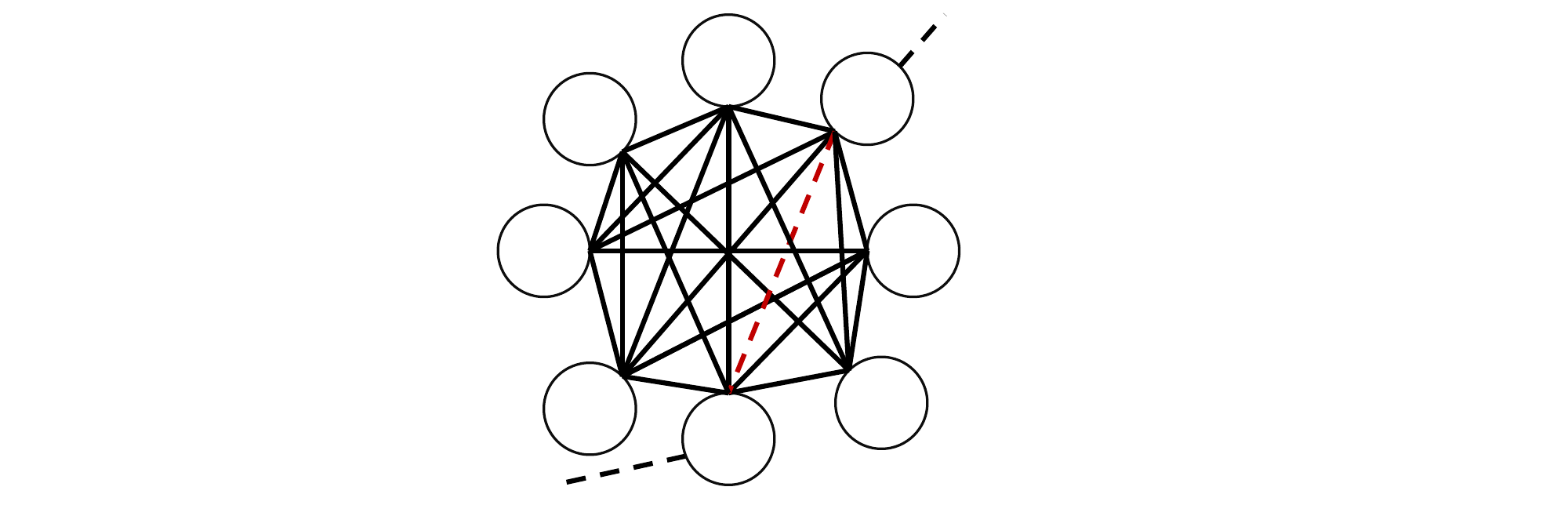}
\caption{This is an example of an almost clique (AC). The depicted nice AC is a clique on $\Delta+1$ nodes with a single missing (red) edge. It is essential that the two nodes incident to the missing edge receive the same color to solve the $\Delta$-coloring problem. All non-nice ACs form proper cliques.}
\label{fig:NiceClique}
\end{figure}

Like in all recent randomized distributed graph coloring algorithms, they divide the graph into sparse and dense parts that are referred to as "almost-cliques" (ACs). Then, they partition the ACs further into different types -- \emph{ordinary, nice, difficult} -- each of which admits a different coloring approach. See \Cref{fig:NiceClique} for an example of an AC.  One challenge is that all these different types of tricky subgraphs may appear in the same graph and close to each other. For this overview it is best to imagine each AC as a proper clique on almost $\Delta$ nodes in which each node has a few external neighbors residing in other ACs and creating lots of dependencies between different ACs.   Thus, their algorithm is fragile with regard to the order in which different types of ACs are colored. 
The starting point of our work is that the core step of their algorithm does not work in low-degree graphs. More detailed, the first step of their algorithm executes \slackgeneration (see \Cref{alg:slackgen} in \Cref{sec:deltaColoring}) on a carefully selected subset of nodes to achieve three objectives: \textbf{a)} giving slack to all sparse nodes, \textbf{b)} providing a slack-toehold\footnote{A slack-toehold for an AC is an uncolored node that can be stalled to be colored later. 
All of its neighbors then lose one competitor for the remaining colors, providing them with temporal slack. 
} for a subclass of the difficult ACs that the authors term "runaway", and \textbf{c)} providing each ordinary clique with a node that has slack. Each of these probabilistic guarantees holds w.h.p.\ as long as $\Delta=\omega(\log^3 n)$. Their proof shows that, in essence, all three cases are LLLs but ones that are far from being simulatable. We discuss our solutions for a)--c),  separately. 

\smallskip

\textbf{Solution for a): } Providing slack to sparse graphs is the main application of the LLL algorithm in \cite{HMP24}. In essence, we adapt their techniques to provide slack to sparse nodes but provide additional guarantees that are needed for other parts of the graph.

\smallskip

\begin{figure}
\centering
\includegraphics[width=0.9\textwidth]{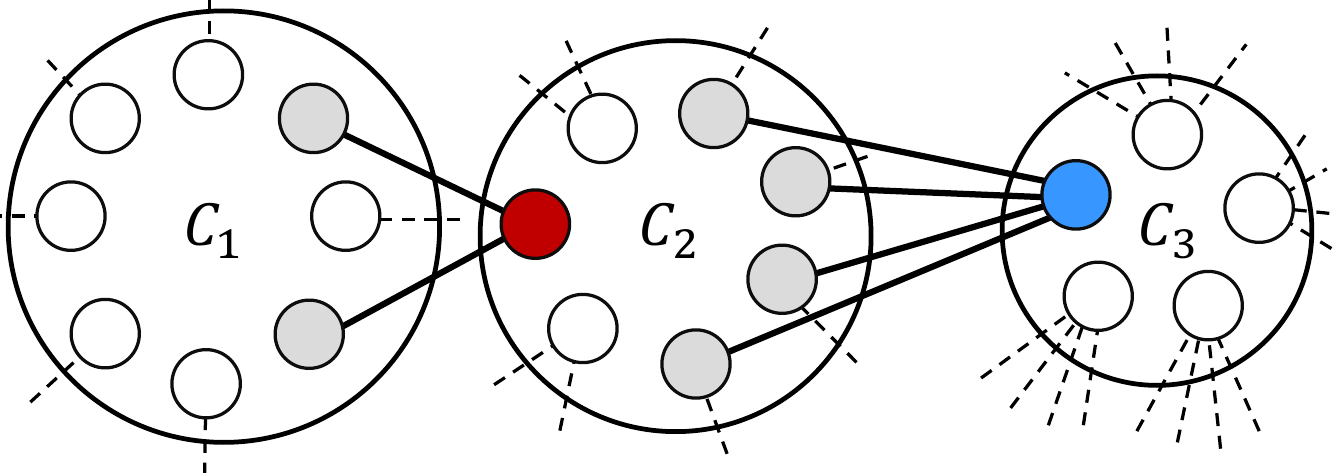}
\caption{\textbf{for part b): }The illustration depicts three difficult cliques of different layers. The external degree of $C_1$ is $1$, the external degree of $C_2$ is $2$ and the external degree of $C_3$ is $4$. $C_1$ has the lowest layer and its special node (the red node) is part of $C_2$. The blue special node of $C_2$ is part of $C_3$. So when we color $C_1$ the red node serves as an uncolored toehold providing slack to two gray nodes of $C_1$. Stalling the coloring of these gray nodes provides slack to the white nodes of $C_1$ so that they can be colored, followed by the gray ones. 
For illustration purposes, we chose $\Delta$ to be $9$, but note that this would actually not classify $C_3$ as a difficult clique. A special node of $C_3$ would need $2e_{C_3}=8$ neighbors in $C_3$, which is impossible due to $C_3$'s size.  }
\label{fig:difficult}
\end{figure}

\textbf{Solution for b):}  For the difficult cliques we propose a solution that eliminates randomness and solely colors all the nodes via a sequence of d1LC instances. See \Cref{fig:difficult} for an illustration of our solution. First, we adjust the classification of difficult almost-cliques from \cite{FHM23}. All nodes in a given  difficult clique have the same external degree. We associate with each such AC $C$ a \emph{special node} $s_C$ on its outside that has many neighbors on the inside (namely, more than twice the external degree of $C$'s nodes).

From here, we assign each difficult clique a \emph{layer} that determines the step in which it gets colored. 
Those with a special node that is \emph{not} contained in another difficult clique are treated separately and assigned to layer $\infty$, to be dealt with at the very end.
The other difficult cliques are assigned to layers indexed by the base-2 logarithm of their external degree. The crucial property that follows is that the cliques in a given layer have their special node in a higher layer. This allows us to color the cliques layer by layer, starting with smaller layers.
The special node $s_C$ is stalled to be colored later, providing a toehold for $C$. This way, we color the cliques and special nodes in all layers besides $\infty$.

This leaves the problem of coloring ACs the $\infty$ layer and their still uncolored special nodes.
In this exposition, we assume that special nodes are not shared by multiple difficult cliques.  In that case, we pair the special node $s_C$ up with some node $u_C\in C$ that is not adjacent to $s_C$ with the objective to \emph{same-color} the nodes: assigning both the same color. This is done via a virtual coloring problem capturing the dependencies between all selected pairs in the participating difficult cliques and the restrictions imposed by already colored vertices of the graph. We show that this virtual coloring instance is indeed a d1LC instance and can be solved efficiently in \CONGEST despite being a problem on a virtual graph. As a result, the clique $C$ obtains an uncolored node $y_C$ that is adjacent to both $s_C$ and $u_C$, has slack due to two same-colored neighbors, and can serve as a toehold for $C$. 

Besides removing the need for randomization to solve the difficult cliques, our classification of difficult cliques also captures significantly more ACs than the definition of difficult cliques in \cite{FHM23}. The additional structure provided to the remaining ACs is exploited down the line in the most challenging part of the algorithm, dealing with the ordinary cliques in part c).

\begin{figure}
\centering
\includegraphics[width=\textwidth]{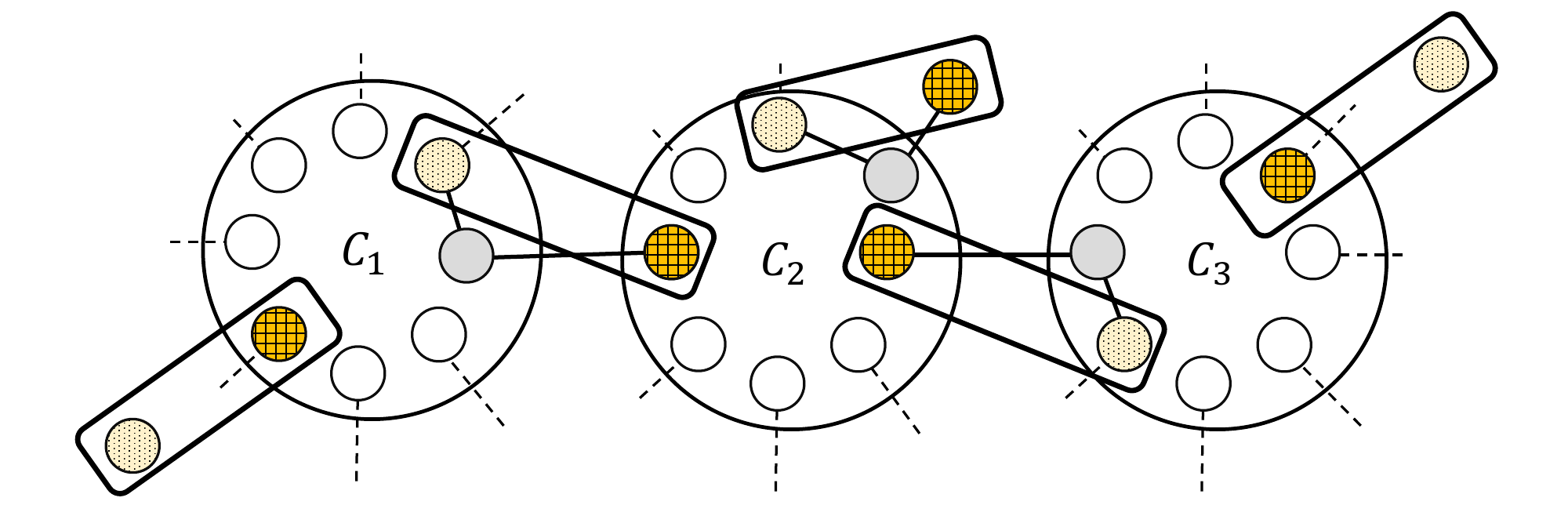}
\caption{\textbf{for Part c): }For the large ordinary cliques, we find triples of nodes consisting of a yellow (striped), a light yellow (dotted), and a gray (solid) node. The two yellow nodes are non-adjacent while the gray node is adjacent to both of them.  The goal is to same-color the pairs of yellow/light-yellow nodes, to which end we form a virtual coloring instance consisting of all pairs and their dependencies. After same-coloring the yellow nodes,
the gray node provides a slack-toehold for the clique. An important aspect is that triples of different ordinary ACs are non-overlapping and no neighborhood of the graph contains too many nodes in such pairs, as otherwise we may run into unsolvable subinstances down the line. We find these triples by a sequence of `simple' LLLs.  }
\label{fig:ordinaryCliques}
\end{figure}

\smallskip

\textbf{Solution for c):}  
The most involved part by far is dealing with case c). We split the ordinary cliques into the small (of size less than $\Delta - \Delta/\poly\log\log(n)$) and large. The small ones can be handled just like the sparse nodes, as one can show that their induced neighborhoods are relatively sparse.
The main effort then is to manually create slack for the large ordinary cliques. For this exposition, it is best to imagine an ordinary clique to be a clique on $\Delta$ nodes in which each node of the clique has exactly one external neighbor that is again a member of a large ordinary clique. See \Cref{fig:ordinaryCliques} for an illustration. 

In order to create slack-toehold in each large AC $C$, we compute a "vee-shaped" triple $(x_C,y_C,z_C)$ of nodes, with $x_C,y_C\in C$ and $z_C\notin C$, but $z_C\in N(y_C)$ and $z_C$ is also a non-neighbor of $x_C$. Then, we set up a virtual list coloring instance with a node for each such pair with the objective to same-color the pairs $(x_C,z_C)$. As we ensure that  $y_C$ is uncolored, it serves as a slack-toehold for the AC. As many of the important ACs can be mutually adjacent, the main difficulty lies in finding non-overlapping triples for the ACs. We ensure this by first computing a suitable candidate set $Z$ from which we then pick the third node $z_C$ of the triple.
Finding the set $Z$ can be modeled as an `easy' LLL fitting the framework of \cite{HMP24}. Finding the node $z_C\in Z$ can also be modeled as a different type of `easy' LLL. In essence, the first LLL is easy (in \CONGEST) as its bad events only consist of simple bounds on the number of neighbors in $Z$. Next, we elaborate on our LLL  for finding $z_C\in Z$ with slightly more detail; due to further technicalities of the existing LLL algorithms from which we spare you in this technical overview, our actual solution differs slightly from the one presented here. 

With a given set $Z$, we model the problem of selecting $z_C\in Z$ as an LLL as follows. Each AC $C$ sends a proposal (to serve as its $z_C$ node) to each outside neighbor inside $Z$ with probability $\poly\log \log n/\Delta$. The proposal is \emph{successful} if no other AC proposes to that node. 
We show that with a constant probability, no other AC proposes to the same node and that this is independent for different nodes in $Z$. 
Since we ensure $C$ has many neighbors in $Z$, we obtain that the probability that none of $C$'s proposals are successful is bounded above by $p=\exp(-\Omega(\poly\log n))$. 
The main benefit is that this LLL and also the LLL for finding the set $Z$ are simple enough to be simulatable (in contrast to  LLLs based on randomized slack generation for those ACs that can be derived from the proofs in \cite{HMP24}). 

Once we have found $z_C$, the structure of large ordinary ACs implies that we can deterministically find the other two nodes $x_C$ and $y_C$ of the triple. Additional complications arise in ensuring that the list coloring instance of the pairs is
a d1LC instance, i.e., that the size of the joint available color palette of $x_C$ and $z_C$ exceeds the maximum degree in the virtual graph induced by the pairs. The last difficulty that appears is solving the d1LC instance, as the bandwidth between the nodes within a pair is very limited and existing d1LC algorithms cannot be run in a black-box manner. 

 \subparagraph{Further related work.} Graph coloring is fundamental to distributed computing as an elegant way of breaking symmetry and avoiding contention, and was, in fact, the topic of the original paper introducing the \LOCAL model \cite{linial92}. 
 There is an abundance of efficient deterministic and randomized $\Delta+1$-coloring algorithms in \LOCAL and \CONGEST for various settings, e.g., \cite{barenboim15,HSS18,FHK,CLP18,BKM20,RG20,MT20,HKMT21,HN21,HKNT22,FK23}. The excellent monograph on distributed graph coloring by Barenboim and Elkin is still a great resource for older results \cite{barenboimelkin_book}. 
 
 There are significantly fewer results for coloring with fewer than $\Delta+1$ colors. 
A \LOCAL algorithm is known for $\Delta-k$-coloring in graphs not containing too large cliques \cite{BamasEsperet19}. An $O(\log\log n)$-round $\Delta$-coloring algorithm in the \LOCAL model is known for trees \cite{CHLPU20}, matching the lower bound \cite{brandt2016LLL} within a constant factor. Additionally, there are works coloring special graph classes such as coloring planar graphs with $6$ or $5$ colors in $O(\log n)$ rounds with a deterministic \LOCAL algorithm \cite{CM19,Postle19}. 

\subparagraph{Outline.} In \Cref{sec:preliminaries}, we define the notion of slack and state required results from prior work on solving d1LC and computing an almost clique decomposition (ACD). In \Cref{sec:deltaColoring}, we present our $\Delta$-coloring algorithm with essentially all proofs. The algorithm consists of $5$ phases and all phases except for Phases~1 (ACD computation) and Phase~2 are deterministic reductions to various d1LC instances. In Phase~2, we provide slack to sparse nodes and the nodes in ordinary cliques; this refers to part a) and part c) described in \Cref{sec:nutshell}. For ease of presentation, the (involved) Phase~2 is presented in two consecutive sections, where we first reduce Phase~2 to solving four different subproblems in \Cref{sec:deltaColoringSmall}  and then solve each of these subproblems via an instance of the constructive \lovasz Local Lemma in \Cref{sec:LLLsubproblems}.

\section{Preliminaries: d1LC, Slack, Almost-Clique Decomposition, Graytone}
\label{sec:preliminaries}
In the $deg+1$-list coloring (d1LC) problem, each node of a graph receives as input a list of allowaed colors whose size exceeds its degree. The goal is to compute a proper vertex coloring in which each node outputs a color from its list. The problem can be solved with a simple centralized greedy algorithm, and it also admits efficient distributed algorithms. 
\begin{lemma}[List coloring \cite{HKNT22,HNT22}]
	\label{lem:listColoring}
	There is a randomized \CONGEST algorithm to $(deg+1)$-list-color (d1LC) any graph in $O(\log^5 \log n)$ rounds, w.h.p.  This reduces to $O(\log^3 \log n)$ rounds when the degrees and the size of the color space is $\poly(\log n)$.
\end{lemma}

The slack of a node (potentially in a subgraph) is defined as the difference between the size of its palette and the number of uncolored neighbors (in the subgraph). 

\begin{definition}[Slack]
    \label{def:slack}
    Let $v$ be a node with color palette $\Psi(v)$ in a subgraph $H$ of $G$. The slack of $v$ in $H$ is the difference $|\Psi(v)| - d$, where $d$ is the number of uncolored neighbors of $v$ in $H$. 
\end{definition}

We use the following helpful terminology.
\begin{definition}[Graytone \cite{FHM23}] 
Consider an arbitrary step of the algorithm. A node is \emph{gray} if it has unit-slack or a neighbor that will be colored in a later step of the algorithm. A node is \emph{grayish} if it is not gray but has a gray neighbor.
A set of gray and grayish nodes is said to be \emph{graytone}.
\end{definition}
Any graytone set can be colored as two d1LC instances: first the grayish nodes and then the gray. 
We emphasize that the graytone property depends on the order in which nodes are processed. It always refers to a certain step of the algorithm in which we color the respective set. Throughout our algorithm we aim at making more and more nodes graytone.

\smallskip

The following construction is central to our approach. 

\begin{restatable}[ACD computation \cite{AKM22,FHM23}]{lemma}{lemACDcomputation}
	\label{lem:acd}
	For any graph $G=(V,E)$, there is a partition (\emph{almost-clique decomposition (ACD)} of $V$ into sets $\Vsparse$ and $C_1, C_2, \ldots, C_t$ such that each node in $\Vsparse$ is $\Omega(\epsilon^2\Delta)$-sparse\ and for every  $i\in [t]$,
    \begin{compactenum}
     \renewcommand{\labelenumi}{(\roman{enumi})}
        \item $(1 - \eps/4)\Delta \le |C_i|\le (1+\eps)\Delta$\ ,
		\item Each $v\in C_i$ has at least $(1-\eps)\Delta$ neighbors in $C_i$:   $|N(v)\cap C_i|\ge (1-\eps)\Delta$\ ,
		\item Each node $u \not\in C_i$ has at most $(1-\eps/2)\Delta$ neighbors in $C_i$: $|N(u)\cap C_i|\le (1-\eps/2)\Delta$.	
	\end{compactenum}
	Further, there is an $O(1)$-round \CONGEST algorithm to compute a valid ACD, w.h.p.
\end{restatable}
We adapt a proof from \cite{FHM23} that, as stated, applies only to the case when $\Delta$ is sufficiently large. Technically, the argument differs only in that we build on \cite{HNT22}
instead of \cite{HKMT21}
in the first step of the argument, where we compute a decomposition with weaker properties. We have opted to rephrase it, given the different constants in the definitions of these works and in order to make it more self-contained.
\begin{proof}[Proof of \Cref{lem:acd}]
    We first use a $O(1)$-round \CONGEST algorithm of [HNT22] 
    to compute a weaker form of ACD with parameter $\eps/4$.\footnote{While such a statement is used in the paper, it is not explicitly stated. Alternatively, we may use an alternative (slower) implementation (Lemma 4.4) in [Flin et al (FGHKN22), arXiv:2301.06457]] that runs $O(\log\log n)$ rounds for $\Delta = \poly(\log n)$ and still suffices for our main result. The slowdown in [FGHKN22] comes from working with sparsified graphs, while a \CONGEST version also runs in $O(1)$ rounds.}
    Namely, it computes w.h.p.\ a partition $(V', D_1, D_2, \ldots, D_k)$ where nodes in $V'$ are $\Omega(\Delta)$-sparse and we have, for each $i\in [k]$:
    \begin{compactenum}
    \renewcommand{\labelenumi}{(\alph{enumi})}
    \item $|D_i|\le (1+\eps/4)\Delta$, and
    \item $|N(v)\cap D_i|\ge (1-\eps/4)\Delta$, for each $v \in D_i$.
    \end{compactenum}    
    What this construction does not satisfy is condition (iii).

    We form a modified decomposition $(\Vsparse, C_1, \cdots, C_k)$ as follows.
    For each $i \in [t]$, let $C_i$ consist of $D_i$ along with the nodes in $V'$ with at least $(1-\eps)\Delta$ neighbors in $D_i$. Let $\Vsparse = V \setminus \cup_i C_i$.
   Observe that the decomposition is well-defined, as a node $u\in V'$ cannot have $(1-\eps)\Delta>\Delta/2$ neighbors in more than one $D_i$. 

   We first bound from above the number of nodes added to each part $C_i$. 
   Each node in $D_i$ has at most $\eps \Delta/4$ outside neighbors, so the number of edges with exactly one endpoint in $D_i$ is at most $\eps\Delta|D_i|/4 \le \eps(1+\eps/4)\Delta^2/4$, using (a) to bound $|D_i|$. Each node in $C_i \setminus D_i$ is incident on at least $(1-\eps)\Delta$ such edges (by definition). 
	Thus, 
 \begin{equation}
  |C_i \setminus D_i| \le  \eps/4 \cdot (1+\eps/4)\Delta/(1-\eps) \le \eps\Delta/2 \ .
   \label{eq:added}
   \end{equation}
   
    Now, (iii) holds since a node outside $C_i$ has at most $(1-\eps)\Delta$ neighbors in $D_i$ (by the definition of $C_i$) and at most $|C_i \setminus D_i| \le \epsilon\Delta/2$ other neighbors in $C_i$ (by \cref{eq:added}). 
    Also, (ii) holds for nodes in $D_i$ by (b) and for nodes in $C_i \setminus D_i$ by the definition of $C_i$.  
    For the lower bound in (i), $|C_i| \ge |D_i| \ge (1-\eps/4)\Delta$, by (b).
    For the upper bound of (i), we have $|C_i|\le |D_i|+|C_i\setminus D_i| \le (1+3\epsilon/4)\Delta$ (by (a) and \cref{eq:added}). 
    
    Finally, the claim about $\Vsparse$ follows from the definition of $V'$, as $\Vsparse \subseteq V'$.
\end{proof}

We say that nodes in $\Vsparse$ are \emph{sparse} and other nodes are \emph{dense}.
It is immediate from \cref{lem:acd} that each dense node has external degree (or neighbors outside its AC) at most $\eps \Delta$ and at most $2\eps\Delta$ non-neighbors in its AC. Also, any pair of nodes in $C_i$ have at least $(1-3\eps)\Delta \ge 3\Delta/4$ common neighbors in $C_i$.

\myparagraph{Notation.}
For a graph $G=(V,E)$ and two nodes $u,v\in V$, let $\dist_G(u,v)$ denote the length of a shortest (unweighted) path between $u$ and $v$ in $G$. For a set $S\subseteq V$ we denote $\dist_G(v,S)=\min_{u\in S}\dist_G(v,u)$. $N(v)$ denotes the set of neighbors of a node $v\in V$.

\section{$\Delta$-Coloring in CONGEST}
\label{sec:deltaColoring}
In this subsection, we prove the following theorem.

\thmDeltaColoring*

The extreme cases of very large $\Delta$ and very small $\Delta$ can be solved in the claimed runtime with prior work \cite{FHM23,MU21}, see the proof of \Cref{thm:deltaColoring} in \Cref{sec:proofThmDeltaColoring}. Here, we present an algorithm for the most challenging regime where $\Delta\in O(\poly\log n)\cap \Omega(\poly\log\log n)$. 

In the extreme case that $\Delta=\omega(\log^{21}n)$, the $\Delta$-coloring algorithm from \cite{FHM23} even runs in $O(\logstar n)$ rounds. 
A lower bound of $\Omega(\log_{\Delta}\log n)$ rounds in the \LOCAL model for the $\Delta$-coloring problem \cite{brandt2016LLL} rules out a $O(\logstar n)$ algorithm for small $\Delta$. Hence, in this section, we aim for an algorithm using $\poly\log \log n$ rounds. In fact, we reduce the $\Delta$-coloring problem to a few list coloring instances and a few LLL instances, each of which we solve in $\poly\log\log n$ rounds.

\subsection{Fine-Grained ACD Partition}
\label{ssec:fineGrainedACPartition}

The following definitions of types of almost-cliques are crucial for all results of the paper. The reader is hereby warned to read them slowly!

\begin{definition}[Types of almost-cliques]
For an AC $C$, let $e_C = \Delta - |C| + 1$. 
An AC is \ouremph{easy} if it contains a non-edge or a node of degree less than $\Delta$. 
A node $v\notin C$ is an \ouremph{intrusive} neighbor of a non-easy $C$ if $v$ has at least $2e_C$ neighbors in $C$.  A non-easy AC is \ouremph{difficult} if it has an intrusive neighbor.
Each difficult AC $C$ arbitrarily selects one of its intrusive neighbors as its \ouremph{special} node $s_C$.
An AC is \ouremph{nice} if it is easy or if it is both non-difficult and contains a special node (necessarily for another AC). 
An AC is \ouremph{ordinary} if it is neither nice nor difficult.    
\end{definition}

Note that all ACs except the easy are proper cliques and all nodes in such a clique $C$ have external degree $e_C$.
We say that a node is ordinary (difficult, nice) if it belongs to an ordinary (difficult, nice) AC, respectively. The difficult ACs are divided into \emph{levels}.

\begin{definition}[Levels of difficult ACs]
\label{def:levels}
 The \emph{maximum level} $\infty$ contains all difficult ACs whose special node is not contained in a difficult AC.  A difficult AC $C$ that is not at the maximum level has \emph{level} $\ell(C) = \lceil \log_2 e_C \rceil$.
\end{definition}
Observe that $\ell(C) \le \log_2 \Delta = O(\log\log n)$ for all difficult ACs. 
\begin{definition}[Node classification]
\label{def:nodeClassification}
\begin{compactenum}
The nodes are partitioned into the following sets:
	\item $\cS$: the set of special nodes that are not in difficult ACs, 
	\item $\cD_\ell$: nodes in difficult ACs of level $\ell$, $\ell \in [\lg \Delta] \cup \{\infty\}$ (might include special nodes),	
	\item $\cN$: nodes in nice ACs, excluding those in $\cS$,
	\item $\cO$: nodes in ordinary ACs, and
    \item $V_{*}$: nodes in $\Vsparse$, excluding those in $\cS$.  
\end{compactenum}
\end{definition}

Our classification is built on \cite{FHM23} but is subtly different and more fine-grained. 
We are driven by a need to limit the reach of probabilistic arguments, being that we are in the challenging sub-logarithmic degree range. Thus, a strictly smaller set of dense nodes (the ordinary) needs probabilistic slack in our formulation. On the other hand, the easy, difficult, and nice definitions are more inclusive here. The difficult ones are here divided into super-constant number of levels, as opposed to only two types in \cite{FHM23}. 

The underlying idea is to ensure that every node gets at least one unit of slack, ensuring that it can be colored as part of a d1LC instance.
Easy nodes have such slack from the start; difficult ones get it from their special nodes (special nodes are used in several different ways to provide slack); sparse and ordinary nodes get it from probabilistic slack generation; and non-easy nice ones get it from same-coloring a non-edge it contains. The most challenging part of the low-degree regime is the probabilistic part. That has guided our definition, resulting in the ordinary ACs being defined as restrictively as possible and, in fact, much more restrictive than the ordinary ACs in \cite{FHM23}.

\subsection{Algorithm for $\Delta$-coloring}
\label{ssec:algorithmLowDegrees}
Our $\Delta$-coloring algorithm consists of the following five phases. 

\begin{algorithm}[H]\caption{$\Delta$-coloring} 
\label{alg:orderedPartition}
	\begin{algorithmic}[1] 
		\STATE Compute an ACD ($\eps=1/172$) and form the ordered partition of the nodes.
        \STATE  Color sparse nodes $V_*$ and ordinary nodes $\cO$ 
		\STATE Color nice nodes $\cN$
		\STATE For increasing $1\le \ell<\infty$ : \\
    \quad Color difficult nodes $\cD_{\ell}$ in level $\ell$
		\STATE Color difficult nodes in $\cD_{\infty}$ and special nodes in $\cS$ 
	\end{algorithmic}
\end{algorithm}
The remainder of the paper describes these phases in detail. 
Only Phases~1 and 2 are randomized. Phase~2 is also the most involved part of our algorithm. For ease of presentation, we defer its details when $\Delta$ is at most logarithmic to \Cref{sec:deltaColoringSmall,sec:LLLsubproblems}. In this section, we present Phase~2 in the case of $\Delta\ge c\log n$ for a sufficiently large constant $c$, where Phase~2 does not require any LLL and which is sufficient to understand how Phase~2 interacts with the remaining phases. The remaining phases are identical in both cases. 

\subsubsection{Phase~1: Partitioning the Nodes}
We first apply \Cref{lem:acd} to compute an ACD for $\eps=1/172$ and break the graph into nice ACs, difficult ACs,  ordinary ACs, and the remaining nodes in $V_*$ according to \Cref{def:nodeClassification}.

\subsubsection{Phase~2: Sparse and Ordinary Nodes ($\Delta\gg\log n$)}
In this subsection, we prove the following lemma. 
\begin{lemma}
\label{lem:sparseOrdinaryLargeDelta}
There exists a $\poly\log\log n$-round \CONGEST algorithm that w.h.p.\ colors the sparse nodes and nodes in ordinary cliques if $\Delta\ge c\log n$ for a sufficiently large constant $c$.
\end{lemma}
\Cref{lem:sparseOrdinaryLargeDelta} essentially follows from the proof of Lemma~3.5 in \cite[arxiv version]{FHM23}. However, as we have changed the definition of ordinary cliques, we spell out the required details.

Slack generation is based on trying a random color for a subset of nodes. Sample a set of nodes and a random color for each of the sampled nodes. Nodes keep the random color if none of their neighbors choose the same color. See \Cref{alg:slackgen} for a pseudocode. If there are enough non-edges in a node's neighborhood, then it probabilistically gets significant slack.

\begin{algorithm}[H]\caption{Phase~2: Coloring Sparse and Ordinary Nodes (when $\Delta\gg\log n$)} \label{alg:sparseOrdinaryEasy}
	\begin{algorithmic}[1] 
		\STATE  Run \textsf{SlackGeneration} on $V_* \cup \cO$
		\STATE Color the remaining ordinary nodes $\cO$ 
		\STATE Color the remaining sparse nodes $V_*$ 
	\end{algorithmic}
\end{algorithm}
\begin{algorithm}[ht]\caption{\slackgeneration} \label{alg:slackgen}
\begin{flushleft}
    \textbf{Input:} $S \subseteq V$
	\begin{algorithmic}[1] 
		\STATE Each node in $v \in S$ is active w.p. $1/20$
		\STATE Each active node $v$ samples a color $r_v$ u.a.r. from $[\chi]$. 
		\STATE $v$ keeps the color $r_v$ if no neighbor tried the same color.  
 	\end{algorithmic}
\end{flushleft}
\end{algorithm}

We also require the following lemma from \cite{FHM23}.
\begin{restatable}[\cite{FHM23}]{lemma}{lemMatchingSlack}
\label{lem:ordinaryMatchingSlack}
Let $C$ be a non-easy AC, $S \subseteq V$ be a subset of nodes containing $C$, 
and $M$ be an arbitrary matching between $C$ and $N(C)\setminus C$.  Then, after $\textsf{SlackGeneration}$ is run on $S$, $C$ contains $\Omega(|M|)$ uncolored nodes with unit-slack in $G[S]$, with probability $1-\exp(-\Omega(|M|))$.
\end{restatable}

There exists a large matching satisfying the hypothesis of  \cref{lem:ordinaryMatchingSlack}, 

\begin{lemma}
\label{lem:M-size}
    For each ordinary AC $C$, there exists a matching $M_C$ between $C$ and $N(C)\setminus C$ of size $2\Delta/5$. 
\end{lemma}
\begin{proof}
 We use the following combinatorial result. 
\begin{restatable}{claim}{obsLargeMatchingExists}
	\label{obs:LargeMatchingExists}
 Let $B=(Y,U,E_B)$ be a bipartite graph where nodes in $Y$ have degree at least $k$ and nodes in $U$ have degree at most $2k$. There exists a matching of size $|Y|/2$ in $B$.
\end{restatable}
\begin{proof}
 Let $M$ be a maximum matching in $B$ and suppose that more than half the nodes in $Y$ are unmatched.
 Let $S$ be the set of nodes reachable from the unmatched nodes $Y \setminus V(M)$.
 Since $M$ has no augmenting path, $S$ contains no unmatched node of $U$.
 All of the $|Y \cap S| \cdot k$ edges incident on $Y \cap S$ have their other endpoint in $U \cap S$. By the degree bound on $U$, there are fewer than $|U \cap S| 2k$ such edges.
Thus, $|Y \cap S| < 2|U\cap S|$.
Every node in $U \cap S$ is matched to a node in $Y \cap S$, while all unmatched nodes in $Y$ are in $Y\cap S$.
Thus, the number of unmatched nodes in $Y$ is at most $|Y \cap S| - |U\cap S| < |U \cap S| \le |M|$.
This is a contradiction, and hence, at least half the nodes in $Y$ are matched.
\renewcommand{\qed}{\ensuremath{\hfill\blacksquare}}
\end{proof}
\renewcommand{\qed}{\hfill \ensuremath{\Box}}

 As $C$ is not easy,	all its nodes have external degree $e_C$, while nodes in $N(C) \setminus C$ are by assumption not intrusive neighbors of $C$, so they have at most $2e_C$ neighbors in $C$. 
 \Cref{obs:LargeMatchingExists} then implies that there exists a matching between $C$ and $N(C)\setminus C$ of size $|C|/2 \ge (1-\epsilon)\Delta/2 \ge 2\Delta/5$.
\end{proof}

The properties of Phase~2 are summarized in the following lemma. 
\begin{lemma}
	\label{lem:smallDegreeSlack}
	If $\Delta\ge c\log n$ for a sufficiently large constant $c$, the following properties hold w.h.p.\ after Step~1 of \Cref{alg:sparseOrdinaryEasy}:
	\begin{compactitem}
		\item[(\dag)] Each sparse node has unit-slack in $G[V_*]$,
		\item[(\dag\dag)] Each ordinary AC has an uncolored unit-slack node in $G[V_*\cup \cO]$.	    
	\end{compactitem}
\end{lemma}
\begin{proof}
We run \slackgeneration on the node set $S=V^*\cup \cO$. Nodes with neighbors outside $V^* \cup \cO$ have slack while the rest of the graph is stalled. We focus on the remaining nodes.
Each sparse node gets the respective slack with probability at least $1-\exp(\Omega(\Delta))$ \cite[Lemma 3.1]{EPS15}, implying (\dag).
By \cref{lem:M-size}, there is a matching between $C$ and $N(C)\setminus C$ of size $2\Delta/5$.
Thus, $(\dag\dag)$ holds with probability at least $1-\exp(-\Omega(\Delta))$, by \Cref{lem:ordinaryMatchingSlack}. 

 Both probabilities become w.h.p.\ guarantees if $\Delta\ge c\log n$ for a sufficiently large constant $c$. For $\Delta\geq \Delta_0$ for a sufficiently large constant $\Delta_0$ we obtain an LLL. 
\end{proof}

\begin{proof}[Proof of \Cref{lem:sparseOrdinaryLargeDelta}]
By \Cref{lem:smallDegreeSlack} w.h.p.\ all sparse nodes become gray as they have unit slack. Also, the unit-slack node in each ordinary AC becomes gray and all other nodes of the AC become grayish as ordinary ACs induce cliques.
This is sufficient to color all nodes with $O(1)$ d1LC instances. 
\end{proof}

\textbf{Forward pointer:} The main difficulty of Phase~2 for smaller values of $\Delta$ is to mimic the properties of \Cref{lem:smallDegreeSlack}. \Cref{sec:deltaColoringSmall,sec:LLLsubproblems} are devoted to ensuring these properties via several LLLs and d1LC instances that can be solved in a bandwidth-efficient manner.

\subsubsection{Phase 3: Nice ACs}
We give a simpler treatment than \cite{FHM23}. 
We want a \emph{toehold} in each nice AC: a node with permanent or temporary slack.
With a toehold, the rest is easy. Namely, ACs have all nodes of internal degree at least $(1-\eps)\Delta$, of which none are colored in previous phases. 
The neighbors of a toehold are gray, and there are at least $(1-\eps/4)\Delta$ of them  by \cref{lem:acd}, all uncolored. The remaining nodes in the AC are then grayish, so the AC is graytone.

Nice ACs come in three types, depending on if they contain a special node, a non-edge, or a degree-below-$\Delta$ node. The first and third types immediately give us a toehold. It remains then to consider nice ACs with a non-edge but with no special node, which we call \emph{hollow}.

For a hollow AC $C$, we identify an arbitrary non-edge $(u_C, w_C)$ and call it \emph{the pair for $C$}. We color the pairs for hollow ACs as a d1LC instance. The two nodes in a pair have at least $\Delta/2$ common neighbors within $C$ and any of them can function as a toehold. It remains to argue that we can find a valid coloring of the pairs efficiently.

\begin{lemma}
	\label{lem:niceAC4}	
	The pairs of hollow ACs can be colored in the \CONGEST model in $O(\log^3 \log n)$ rounds.
\end{lemma}
\begin{proof}
As the nodes of a hollow $C$ were uncolored, the only nodes that can conflict with the coloring of the pair are the at most $2 \cdot \eps \Delta \le \Delta/2$ external neighbors. The $\Delta+1$ colors we have to work with significantly exceed that. Thus, the pairs are $deg+1$-list colorable.

 Both nodes of the pair $(u_C, w_C)$ have
	at least $(1-\epsilon)\Delta$ neighbors in $C$, so they have at least $(1-\epsilon)\Delta - (|C| - (1-\epsilon)\Delta) > (1-3\epsilon)\Delta \ge \Delta/2$ common neighbors in $C$.
	They provide the bandwidth to transmit to one node all the colors adjacent to the other node.  Also, all messages to and from $u_C$ \emph{vis-a-vis} its external neighbors can be forwarded in two rounds. Hence, we can simulate any \CONGEST coloring algorithm on the pairs with $O(1)$-factor slowdown; in particular, we can simulate the algorithm from \Cref{lem:listColoring}.
\end{proof}

\subsubsection{Phase 4: Difficult ACs in a Non-Maximum Level}
By \Cref{def:levels}, the special node $s_C$ of any difficult AC $C$ at a level other than $D_{\infty}$ is contained in another difficult AC $C'\ne C$. The next lemma shows that the level of $C'$ must be strictly larger than the level of $C$, which allows us to color $C$ fast while $C'$ remains uncolored.
\begin{claim} 
\label{claim:specialHigherLevel}
For an AC $C$ with $\ell(C)<\infty$, let $C'$ be the difficult AC that contains the special node $s_C$.
  Then we have $\ell(C) < \ell(C')$.
\end{claim}
\begin{proof}
The special node $s_C$ has external degree of at least $2e_C$ as it is connected to at least $2e_C$ nodes of $C$ that do not lie within $C'$. Hence, we obtain that the external degree $e_{C'}$ in AC $C'$ is at least $e_{C'}\ge 2e_C$, so $\ell(C')>\ell(C)$. 
\end{proof}

We color all ACs of a level in parallel, in increasing order of levels.
Due to the previous claim, the special node of an AC is contained in a difficult clique in a larger level or not contained in a difficult clique at all. Hence, the special node is uncolored when the clique is processed.  So, when processing some level $1\le i\le O(\log \log n)$, we color all nodes in ACs of that level, but we do not color their respective special nodes. Thus, the respective special node provides a toehold for the respective clique.

\subsubsection{Phase 5: Difficult ACs in the Maximum Level}
The maximum level is processed last and differently from the other levels. By definition, the special node $s_C$ of an AC in $\infty$ level is not contained in a difficult AC.
Also, all nodes in $D_\infty$ and their special nodes are still uncolored at the beginning of this phase. 

The algorithm has four steps: (1) Form pairs of selected non-adjacent nodes, (2) Color the nodes in each pair consistently, (3) Graytone color the remaining nodes of the AC, and (4) Color the special nodes $\mathcal{S}$. We explain each step in detail. 

First, we form the following pairs. For each special node $s_C$ that is special for only one AC $C$ at level $\infty$: Form a \emph{type-1} pair $T_s=(s_C,u_C)$ with a non-neighbor of $s_C$ in $C$. For each special node $s$ that is special for more than one ACs at level $\infty$, form a \emph{type-2} pair $T_s=(w_1,w_2)$, where $w_1$ and $w_2$ are arbitrary non-adjacent nodes in two of the ACs for which $s$ is special. Let $\cE$ be the set of the latter special nodes.
\begin{claim}
	The pairs can be properly formed.
\end{claim}
\begin{proof}
	\textit{Type-1:} An (uncolored) non-neighbor $u_C$ of $p_C$ exists as $p_C$ can have at most $(1-\eps/2)\Delta$ neighbors in $C$ by \Cref{lem:acd} (4), but the AC $C$ has at least $(1-\eps/4)\Delta$ vertices. 
	
	\textit{Type-2:} Let $C_1$ and $C_2$ be two ACs at level $\infty$ for which $s$ is special, where $e(C_1) \le e(C_2)$. By definition, $s$ has at least $2e(C_1)$ ($2e(C_2)$) neighbors in $C_1$ ($C_2$), respectively.  Pick $w_1$ to be any neighbor of $s$ in $C_1$. Node $w_1$ has at most $e(C_1)$ neighbors in $C_1$. Thus, there are at least $2e(C_2) - e(C_1) > 0$ nodes in $C_2$ that are neighbors of $s$ and non-neighbors of $w_1$, and we can pick any such node as $w_2$. 
\end{proof}

\begin{lemma}
\label{lem:lastLevelColoring}
	Coloring the pairs is a $(deg+1)$-list coloring instance that can be solved in $\poly\log\log n$ rounds in \CONGEST, w.h.p. 
\end{lemma}
\begin{proof}
	\textit{Type-1 pair $T=\{s_C,u_C\}$, $s_C\notin C$, $u_C\in C$:}  
	We say that a node \emph{conflicts} with the pair $\{s_C,u_C\}$ if the node is already colored or is contained in an adjacent pair of the same phase. As $C$ does not contain a special node, $u_C$ is the only node of $C$ participating in the phase and all other nodes of $C$ are still uncolored. The node $u_C$ can only be adjacent to $e_C$ conflicting nodes as it has external degree at most $e_C$. As $s_C$ has at least $2e_C$ neighbors in $C$, it can conflict with at most $\Delta-2e_C$ nodes.  Thus, the pair conflicts with at most $e_C+\Delta-2e_C=\Delta-e_C$ nodes, which is less than $\Delta$, the number of colors initially available. Thus, the problem of coloring such pairs is a $(deg+1)$-list coloring problem. 
	
	\textit{Type-2 pair $T=\{w_1,w_2\}$:} Each such pair $(w_1,w_2)$ is adjacent to at most $e(C_1)+e(C_2) \le 2\epsilon \Delta$ nodes in other ACs. Further, all nodes in the ACs $C_1$ and $C_2$ are still uncolored, so both nodes have at least $(1-2\epsilon)\Delta$ colors in their palette, and each pair is adjacent to at most $2\eps\Delta$ other pairs or already colored neighbors, that is, the palette exceeds the degree.
	
	\textit{\CONGEST Implementation.}
	A type-1 pair has at least $e(C)$ common neighbors (the special node $s_C$ has $2e(C)$ neighbors inside the clique by its definition that are all connected to $u_C$), which suffices to communicate the colors and all messages of external neighbors of $u_C$ to $s_C$ ($u_C$ has at most $e_C$ external neighbors). 
	Hence, the coloring can be achieved in \CONGEST.

	Let $s$ be the common special node of a type-2 pair $\{w_1,w_2\}$ and let $C_1$ and $C_2$ be the respective cliques. For $i=1,2$ the node $w_i$ has at most $e_{C_i}$ outside neighbors and $s$ has $2e_{C_i}\ge e_{C_i}$ neighbors in $C_i$, denote these by $X_i$. We simulate the pair by $s$. The node $w_i$ can forward all initial colors of outside neighbors as well as all messages from them to $s$ by relaying them through $X_i$. 
\end{proof}

After coloring the pairs, each difficult AC $C$ has a node with unit-slack in $G[V\setminus \cE]$, either because the clique contains an uncolored node with two neighbors appearing in a consistently colored type-1 pair $T=\{s_C,u_C\}$, or because it contains an uncolored node with a neighbor in $\cE$. In the former case, the uncolored node exists because $s_C$ has at least one neighbor in $C$ that is also a neighbor of $u_C$. In the latter case, the special node $s$ with type-2 pair $T=\{w_1,w_2\}$ has by definition further neighbors besides $w_1$ and $w_2$ in each clique that are all uncolored.

Thus, we color all nodes in difficult cliques via the graytone property. At the end, we color the nodes in $\cE$, which have unit-slack as they are adjacent to a type-$2$ pair.

\subsection{Proof of Theorem~\ref{thm:deltaColoring}}
\label{sec:proofThmDeltaColoring}

\begin{proof}[Proof of \Cref{thm:deltaColoring}]
There are five cases, depending on the relation of $\Delta$ and $n$. Generally, we use \Cref{lem:listColoring} to solve d1LC instances in $\poly\log\log n$ rounds. Whenever the d1LC instances require additional arguments to be solved in the respective time, e.g., because they are defined on a virtual graph, we reason their runtime when they are introduced. 
\begin{itemize}
\item If $\Delta=\omega(\log^4 n)$, we use the algorithm from \cite{FHM23} to $\Delta$-color the graph. 

\item For $c\log n\leq \Delta=O(\log^4 n)$ for a sufficiently large constant $c$, 
the result follows by executing \Cref{alg:orderedPartition} with the arguments of this section. Phases $1$--$3$ only require $O(1)$ rounds and a constant number of d1LC instances.  In Phase~4, we iterate through the $O(\log\Delta)=O(\log\log n)$ levels and solve a constant number of d1LC instances for each level. Phase~5 can be executed in $\poly\log\log n$ time by \Cref{lem:lastLevelColoring}.

\item When $\poly\log\log n \le \Delta \le c\log n$, we
use \Cref{alg:orderedPartition} from this section and replace Phase~2 with \Cref{alg:sparseOrdinaryHard} (presented in \Cref{sec:deltaColoringSmall}) whose correctness and runtime we prove in \Cref{sec:deltaColoringSmall,sec:LLLsubproblems}. 
\item If $\Delta_0\le \Delta\le \poly\log\log n$, we use the algorithm of this section together with the LLL representation from the proof of \Cref{lem:smallDegreeSlack}. The LLL can be solved with the \CONGEST LLL solver of \cite{MU21} in $\poly\Delta\poly\log\log n=\poly\log\log n$ rounds. Here, $\Delta_0$ is a sufficiently large constant such that the LLL guarantees from \Cref{lem:smallDegreeSlack} hold. 
\item If $3\le \Delta\le \Delta_0$, that is, for constant $\Delta$, there is an existing algorithm from \cite{MU21}.
\end{itemize}
 
 In all cases, the algorithm runs in $\poly\log\log n$ rounds. 
\end{proof}
\section{Phase~2 ($\Delta=O(\log n)$): Sparse Nodes and Ordinary Cliques}
\label{sec:deltaColoringSmall}
In this section, we deal with Phase~2 for the most challenging regime of 
$\Delta\in O(\log n)\cap \Omega(\poly\log\log n)$. 
The following lemma follows from all proofs in this section, together with \Cref{lem:slackSparseSmallOrdinary,lem:usefulEdges,lem:tripleLemma} all proven in \Cref{sec:LLLsubproblems}.

\begin{lemma}[Phase~2]
\label{lem:DeltaColoringIntermediate}
There exists a $\poly\log\log n$-round \CONGEST algorithm that w.h.p.\ color the sparse nodes and nodes in ordinary cliques if $\log^{10}\log n \le \Delta\le O(\log n)$.
\end{lemma}

We first give high-level ideas of our method. 
We divide the ordinary cliques into the \emph{small}, of size at most $\Delta(1-1/(10 \log^3 \log n))$, and the \emph{large}.
Nodes in small ordinary cliques have significant sparsity (i.e., non-edges in their induced neighborhood), which means that the one-round procedure of trying a random color has a good probability of successfully generating slack. The natural LLL formulation of that step is therefore well-behaved enough that it can be solved fast in \CONGEST with a few additional tweaks, see \Cref{sec:slackSparseSmallOrdinary}.
Large nodes need a different approach.

For each large AC, we produce unit slack for a single node. See \Cref{fig:ordinaryCliques} for an illustration of the process we will describe. We identify for each such AC a triplet of nodes $(x,y,z)$ with the objective to color $x$ and $z$ with the same color, while $y$ remains uncolored. This way, $y$ receives unit slack, which gives us a toehold to color the whole AC. 

Computing such triplets is non-trivial. We do so by breaking it into three steps, each solvable by a different LLL formulation. In brief, we first compute a set $Z$ of candidate $z$-nodes; next partition $Z$ into two sets; and then select the actual $z$-nodes to be used from these two sets. The split of $Z$ into two sets is required to make the process of finally finding the $z$-nodes fit the LLL solver from \Cref{thm:LLLTwoSets}. The properties of the set $Z$ imply that it is then much easier to identify compatible $x$- and $y$-nodes, and once we find such triplets, we set up a virtual coloring instance for same-coloring $x$- and $z$-nodes in each triple. We show that this instance is d1LC and can be solved with low bandwidth despite being defined on a virtual graph. This provides a slack-toehold to the $y$-node of each triple and the coloring can be extended via d1LC instances to the whole instance.

\subparagraph{Algorithm.}

The first step of the algorithm is to compute a large matching $M_C$ between each ordinary clique $C$ and $N(C) \setminus C$ in parallel. 
We then classify the ordinary cliques as follows.
Fix the parameter $q(n)=10\log^3\log n$ throughout this section.

\begin{definition}[Small, Large, Unimportant and Important Ordinary cliques.] 
\label{def:ordinarytypes}
An ordinary AC is \emph{large} if it contains more than $\Delta-\Delta/q(n)$ nodes, and \emph{small} otherwise. 
A large AC is \emph{important} if $|(V(M_C)\setminus C) \cap \mathcal{O}_l|\ge \Delta/12$, and \emph{unimportant} otherwise. 
\end{definition}

We say that a node is small/large/important/unimportant if it belongs to an AC of the corresponding type.
Let $\mathcal{O}_i$, $\mathcal{O}_u$, $\mathcal{O}_l = \mathcal{O}_i \cup \mathcal{O}_u$, and $\mathcal{O}_{s}$ be the set of important, unimportant, large, and small nodes, respectively.

Next, we present our full solution. The algorithm has the following steps, which are explained in detail below.

\begin{algorithm}[H]\caption{Phase~2: 
Coloring Sparse and Ordinary Nodes ($\Delta=O(\log n)$)} \label{alg:sparseOrdinaryHard}
	\begin{algorithmic}[1] 
 \STATE Step~0: For each ordinary AC $C$ in parallel, compute a matching $M_C\subseteq C\times (N(C)\setminus C)$. 
 Classify ordinary ACs into important, unimportant, and small ACs. \\ 
		\STATE Step~1: Generate slack for sparse and small nodes (via LLL, \Cref{sec:slackSparseSmallOrdinary})
		\STATE Step~2: Compute candidate sets $Z=Z_1\cup Z_2\subseteq \mathcal{O}_l$ (via LLL, \Cref{sec:findingZ})
		\STATE Step~3: Form triples $(x_C,y_C,z_C)\in C\times C\times Z$ (via LLL, \Cref{sec:TripleForming}) 
		\STATE Step~4: Same-color $(x,z)$-pairs via virtual coloring instance 
		\STATE Step~5: Color the remainder of $V^* \cup \mathcal{O}$  (via d1LC instances). 
	\end{algorithmic}
\end{algorithm}

\subparagraph{Step 0: Classifying ACs and computing matchings.}
We compute a matching $M_C$ for each ordinary clique $C$ between the vertices in $C$ and the ones in $N(C)\setminus C$. 
We use a 2.5-approximate algorithm of \cite{Fischer17} running in $O(\log^2 \Delta+\logstar n)=O(\log^2\log n)$ rounds, obtaining that $|M_C| \ge (2\Delta/5)/2.5 = \Delta/10$, using \cref{lem:M-size}.

We view the edges of $M_C$ as being directed arcs with a head in $C$ and tail in 
$V \setminus C$. 
Each AC can determine its size and the size of $V(M_C) \cap \mathcal{O}_l$ in $O(1)$ rounds and hence the classification of \cref{def:ordinarytypes} can be computed in $O(1)$ rounds. 

\subparagraph{Step 1: Slack for sparse and small nodes.} In this step, we create slack for sparse nodes and all nodes in $\mathcal{O}_s$. 
The key property of small nodes is that they are 
relatively sparse (with many non-edges in their neighborhoods), so randomly trying colors is likely to produce slack. That leads to an LLL formulation that we can make simulatable and can therefore implement in \CONGEST.

The properties are summarized by the following lemma.
Besides providing slack to all sparse nodes and the nodes in small ordinary ACs, it also guarantees that each neighborhood (and hence also each AC) does not have too many nodes colored and that the matching $M_C$ of each AC does not get too many nodes colored. The proof is in \cref{sec:slackSparseSmallOrdinary}.

\begin{restatable}{lemma}{lemSlackSparseSmallOrdinary}
    \label{lem:slackSparseSmallOrdinary}
Assume that we are given a matching $M_C$ of size at least $\Delta/10$ between $C$ and $N(C)\setminus C$ for each ordinary AC $C$.  
    There is a $\poly\log\log n$-round (LLL-based) \CONGEST algorithm that w.h.p.\ colors a subset $S\subseteq V^*\cup \mathcal{O}$ and ensures that:
    \begin{compactenum}
        \item Each uncolored node in $V^*\cup \mathcal{O}_s$ has unit-slack in $G[V^*\cup \mathcal{O}]$.
        \item In each of the following subsets, at most $O(\log^4\log n \cdot \log\Delta)$ nodes are colored: $N(v)$ for each $v\in V^*\cup \mathcal{O}$
        and  $V(M_C)$ for each AC $C$ .
    \end{compactenum}
\end{restatable}

\subparagraph{Step 2: Compute triple candidate set via LLL.} 
Let $X=\mathcal{O}_l\setminus \{v\in \mathcal{O}_l: \text {$v$ colored in Step~1}\}$. 

The goal of this step is to compute two disjoint sets $Z_1, Z_2$ of uncolored nodes 
such that each important AC has sufficiently many matching edges satisfying the following definition of usefulness. 

\begin{definition}[useful edge]
Given a subset $Z \subseteq X$ and important AC $C$, 
a matched arc $\overrightarrow{vu} \in M_C$
is \emph{useful} for $C$ if 
$v \in (X \setminus Z)$ and $u \in Z$.
Refer to $Z$ as the \black nodes and $X\setminus Z$ as the \white nodes. An edge is \emph{\white} if both endpoints are \white. 
\end{definition}
An arc $\overrightarrow{vz}$ cannot be useful for the AC containing $v$; only the one containing $z$.

Formally, Step~2 provides the following lemma that we prove in \cref{sec:findingZ}.
For an AC $C$ and set $Z$, let $U(C,Z)$ denote the arcs of $M_C$ with one endpoint in $Z$ (and the other in $C$). 

\begin{restatable}{lemma}{lemUsefulEdges}
\label{lem:usefulEdges} Let $q= 1/30$. 
There is a $\poly\log\log n$-round (LLL-based) \CONGEST algorithm computing disjoint subsets $Z_1, Z_2\subseteq \mathcal{O}_l$ satisfying the following properties,  w.h.p.:
\begin{compactenum}
  \item $|U(C,Z_i)| \ge q^2(1-q)^3\Delta/60$,
for $i=1,2$ and for each important AC $C$, and
  \item $|(Z_1\cup Z_2)\cap N(v)|\le \Delta/10$, for all $v\in \mathcal{O}$.
\end{compactenum}
\end{restatable}

\subparagraph{Step 3: Forming triples via LLL.}
The goal of this step is to compute a triple $(x_C,y_C,z_C)\in C\times C\times Z$ of nodes that satisfy the conditions of the next lemma.  These triple nodes are distinct for different ACs.

\begin{restatable}{lemma}{lemTripleLemma}   
\label{lem:tripleLemma}
Given sets $Z_1, Z_2\subseteq \mathcal{O}_l$ with the properties as in \Cref{lem:usefulEdges}, 
there is a $\poly\log\log n$-round (LLL-based) \CONGEST algorithm that computes for each large important AC $C$ a triple $(x_C,y_C,z_C)$ of uncolored nodes such that w.h.p.:
\begin{compactenum}  
    \item $x_C, y_C\in C$ and $z_C\notin C$, 
    \item $y_C x_C$, $y_C z_C \in E$, $x_C z_C \not\in E$ ($x_C$ and $z_C$ are non-adjacent; $y_c$ is adjacent to both $x_C$ and $z_C$) and
    \item the graph induced by $\{z_C : C \text{ is important}\}$ has maximum degree $\le \Delta/10$. 
\end{compactenum}
\end{restatable}

We model the problem of selecting $z_C$ for each important AC $C$ as a disjoint variable set LLL.
The proof of the lemma is given in \cref{sec:TripleForming}.

\subparagraph{Step~4: Same-coloring $(x_C,z_C)$ pairs.}
Given a triple ($x_C,y_C,z_C)$, we will create a toehold for the AC $C$ at $y_C$ by coloring its non-adjacent neighbors $x_C$ and $z_C$ with the same color. 

Let $H_P$ ($P$ for pair) be the virtual graph consisting of one vertex for each pair $(s_C,z_C)$ and an edge between two pairs $(s_C,z_C)$ and $(s_{C'},z_{C'})$ if there is any edge in $G$ between $\{s_C,z_C\}$ and $\{s_{C'},z_{C'}\}$. The  \emph{list of available colors} $L((s_C,z_C))$ consists of all colors that are not used by the already colored neighbors in $G$ of $s_C$ and $z_C$. 

\begin{lemma}
\label{lem:HPDelta}
The maximum degree $\Delta_{H_P}$ of $H_P$ is upper bounded by $\Delta/9$.
\end{lemma}
\begin{proof}
    By \Cref{lem:tripleLemma}, each node has at most $\Delta/10$ neighbors in $Z$. Define the set $X'=\{x_C: C\text{ is an important AC}\}$. As $X'$ contains at most one node per AC, the number of neighbors that a node in $\mathcal{O}_l$ can have in $U$ is upper bounded by its external degree plus $1$, which is upper bounded by $\Delta/q(n)+1$. Thus, the maximum degree  $\Delta_{H_P}$ of the virtual graph $H_P$ is at most $\Delta/10+\Delta/q(n)+1\le \Delta/9$ for sufficiently large $n$.
\end{proof}

\begin{lemma}
\label{lem:tripleColoring}
Coloring $H_P$ -- i.e., same-coloring the pairs -- is a $deg+1$-list coloring instance.
\end{lemma}
\begin{proof}
By \Cref{lem:HPDelta} we obtain $\Delta_{H_P}\le \Delta/9$.       As we colored at most $x=O(\log^5\log n)$ vertices in each neighborhood in Step~1,  the list of available colors of each pair has at least $\Delta-2x\gg \Delta/9=\Delta_{H_P}$ colors available in their joint list. Hence, we obtain a $deg+1$-list coloring instance.  
\end{proof}

\myparagraph{\CONGEST implementation.}
Our algorithm is based on the $deg+1$-list coloring algorithm from \cite{ghaffari19_MIS,BEPSv3}. Before we show how to color the nodes in $H_P$, we need to define a slow (it takes $O(\log n)$ rounds) randomized algorithm. 
The algorithm is used in our analysis and it works as follows. In each iteration, each uncolored pair executes the following procedure that may result in the pair to try to get colored with a color or to not try a color (also see \Cref{alg:randomizedPairColoring} for pseudocode of the algorithm). Throughout the algorithm, nodes $x_C$ and $z_C$ maintain lists $L(x_C)$ and $L(z_C)$ consisting of all colors not used by
their respective neighbors in $G$.
Then, in one iteration node $x_C$ selects a color $c$ u.a.r.\  from its list of available colors $L(x_C)$, and sends it to the other endpoint through node $y_C$. The other endpoint $z_C$ checks whether $c \in L(z_C)$; if so, both nodes agree on trying color $c$, and the color is sent to their neighbors. If no incident pair tries the same color, the pair gets permanently colored with the color. Lastly, both nodes individually update their lists by removing colors from adjacent vertices that got colored from their respective list. There is no explicit coordination between the two vertices in maintaining a joint list of available colors. 
\begin{algorithm}[h]\caption{Randomized Pair Coloring} \label{alg:randomizedPairColoring}
	\begin{algorithmic}[1] 
		\STATE Each node $x_C$ selects a color $c$ u.a.r.\  from $L(x_C)$ and sends $c$ to $z_C$
        \STATE If $c\in L(z_C)$ then TryColor(c)
        \STATE Update lists $L(x_C) \leftarrow L(x_C)\setminus \{c(v) : v\in N_G(x_C)\} $ and $L(z_C) \leftarrow L(z_C)\setminus \{c(v) : v\in N_G(z_C)\} $
 	\end{algorithmic}
\end{algorithm}

The next lemma shows that each pair gets colored with constant probability. 
\begin{lemma} 
\label{lem:HPcolorTrialSuccess}
Consider an arbitrary iteration of \Cref{alg:randomizedPairColoring} and an arbitrary pair $(x_C,z_C)$ for a hiding AC $C$ that is uncolored at the start of the iteration. Then, we have 
\begin{align}
    \Pr((x_C,z_C)\text{ gets colored in the iteration})\ge 1/2~.
\end{align}
The bound on the probability holds regardless of the outcome of previous iterations. 
\end{lemma}
\begin{proof}
    Note %
    \footnote{The constants in this proof are not chosen optimally in order to improve readability. } 
    that throughout the execution of \Cref{alg:randomizedPairColoring} the respective lists of nodes $x_C$ and $z_C$ are always of size at least $\Delta-\Delta_{H_P}-\Omega(\log^5\log n)\ge 4\Delta/5$ as $\Delta = \omega(\log^{5}\log n)$ and $\Delta_{H_P}\le \Delta/9$, by \Cref{lem:HPDelta}. Note, that both nodes keep their individual list of available colors in which they only remove the colors of immediate neighbors in $G$ from the list of available colors. Thus, at all times we have $|L(x_C)|\cap  L(z_C)|\ge 3\Delta/5$. Let $X$ be the set of colors tried by one of the $\Delta_{H_P}\le \Delta/9$ pairs incident to $(s_C,z_C)$ in the current iteration. We obtain $|(L(s_C)\cap L(z_C))\setminus X|\ge \Delta/2$. As these colors are at least half of $L(x_C)$'s palette, the probability that the pair $(x_C,z_C)$ gets colored is at least $1/2$. 
\end{proof}

\begin{lemma}
There is a randomized $\poly\log\log n$-round \CONGEST algorithm that w.h.p.\ colors the pairs of $H_P$.
\end{lemma}
\begin{proof}
Consider the well-understood color trial algorithm in which nodes repeatedly try a color from their list of available colors, keep their color permanently 
if no neighbor tries the same color, and remove colors of permanently colored neighbors from their list of available colors. 
It is known that this algorithm colors each node with a constant probability in each iteration \cite{BEPSv3,johansson99}. Thus, it requires $O(\log n)$ rounds to color all vertices of a graph. The shattering-based \CONGEST algorithm from \cite{ghaffari19_MIS} for d1LC runs in $\poly\log\log n$ rounds. It requires three subroutines: a) A color trial algorithm like the one from \cite{BEPSv3,johansson99}, b) a network decomposition algorithm that can run on small subgraphs (the ones in \cite{RG20,MU21,MPU23} do the job), and c) the possibility to run $O(\log n)$ instances of the color trial algorithm in parallel. In our setting we want to solve the same problem, but on the virtual graph $H_p$ while the communication network is still the original graph $G$. The subroutine for part b) can be taken from prior work as the same issue is dealt with formally in \cite{MU21,MPU23,HMP24}. We refer to these works for the details and also the definition of a network decomposition. Let us sketch the main ingredient for the informed reader. Instead of computing a network decomposition of small subgraphs of $H_P$, the subgraphs are first projected to $G$, and a network decomposition of $G$ is computed afterwards. This only requires an increased distance between clusters such that the preimage of the decomposition induces a proper network decomposition of $H_P$. 

For ingredients a) and c), we observe that Ghaffari's algorithm only requires the following properties for the color trial algorithm: i) one iteration can be executed in constant time and with $\poly\log\log n$ bandwidth, allowing to execute $O(\log n)$ instances in parallel in the \CONGEST model, and ii) each node gets colored with a constant probability in each iteration. Thus, we can replace the color trial algorithm with the color trial algorithm for $H_P$ given in \Cref{alg:randomizedPairColoring}. We have already argued that it can be implemented with $\poly\log\log n$ bandwidth showing $i)$ and \Cref{lem:HPcolorTrialSuccess} provides its constant success probability for ii). 
\end{proof}

 \subparagraph{Step 5: Completing the coloring.}

To finish the coloring, we first color the unimportant nodes and then the important, small, and sparse nodes.

\begin{lemma}
    Unimportant nodes are graytone as long as the other ordinary nodes (small, sparse, important) are inactive.
\end{lemma}
\begin{proof}
    The only steps so far in which we colored vertices are Steps~1 and 4. In Step~1 we color at most $O(\log^5\log n)$ vertices per AC and per matching $M_C$ of each ordinary AC $C$. In Step~4 we only color (a subset of) the vertices in $Z$ and one vertex per important AC (the vertex $x_C$ for AC $C$). As $|Z\cap C|\le \Delta/10$, we color at most $\Delta/10+O(\log^5\log n)\le \Delta/9$ vertices in each unimportant AC. 
    
    Fix some unimportant AC $C$. Recall that the algorithm of \cite{Fischer17} finds a 2.5-approximate matching, which by \cref{lem:M-size} implies that $|M_C|\ge \Delta/10$. As an unimportant AC has fewer than $\Delta/12$ nodes in $(V(M_C)\setminus C)\cap \mathcal{O}_l$, we obtain that $V(M_C)\setminus C$ contains at least $\Delta/10-\Delta/12 = 7\Delta/60$ nodes that are not contained in $\mathcal{O}_l$.  By \Cref{lem:slackSparseSmallOrdinary}, at most $O(\log^5\log n)$ of these get colored in Step~1; denote the uncolored nodes of these by $S$ and let $S'=N(S)\cap C$. By the earlier argument, at most $\Delta/9$ nodes of $S'$ are already colored, that is, there exists some $v\in S'$ that is still uncolored and has an uncolored neighbor $u\notin \mathcal{O}_l$. As $u$ is stalled to be colored later, $v$ is gray and other nodes of the AC are grayish. 
\end{proof}

\begin{lemma}
    Small, sparse, and important nodes are graytone.
\end{lemma}

\begin{proof}
    By \Cref{lem:slackSparseSmallOrdinary}, each small or sparse node has slack in $G[V^*\cup \mathcal{O}]$ and is therefore gray (and stays gray until colored).

    For an important AC $C$ with triple $(x_C,y_C,z_C)$, the node $y_C$ is gray as $x_C$ and $z_C$ are same-colored. Hence, the remaining uncolored nodes of $C$ are either already colored or graytone as they are adjacent to $v$. 
\end{proof}

\section{Solving Subproblems of Phase 2 via LLL}
\label{sec:LLLsubproblems}
We show how the probabilistic subproblems of \Cref{sec:deltaColoringSmall} can be solved via a fast LLL algorithm. We show for all four problems that they can be captured with the \CONGEST framework of \cite{HMP24}.
We start by reviewing the framework of \cite{HMP24} and then solve each of the subproblems in respective subsections.

\subsection{Framework for LLL in CONGEST}
\label{sec:LLLdefinitions}
In this section, we present \CONGEST model LLL solvers from \cite{HMP24}. The definitions, theorems, and selected textual excerpts in this section have been sourced from \cite{HMP24}. 

\subparagraph{Constructive \lovasz Local Lemma (LLL).} 
An instance $\cL=(\cV,\cB)$ of the \emph{distributed \lovasz local lemma (LLL)} is given by a 
a set $\cV=\{x_1,\ldots,x_{k_{\cV}}\}$ of independent random variables and a
family $\cB$ of "bad" events $\{\cE_1,\ldots,\cE_{k_{\cB}}\}$ over these variables. Let $\vbl(\cE)$ denote the set of variables involving the event $\cE$ and note that $\cE$ is a binary function of $\vbl(\cE)$.
The \emph{dependency graph} $\cH_{\cL}=(\cB, F)$ is a graph with a vertex for each event and 
an edge $(\cE,\cE') \in F$ whenever 
$\vbl(\cE)\cap \vbl(\cE')\ne\emptyset$. The \emph{dependency degree} $d = d_{\cL}$ is the maximum degree of $H_{\cL}$. 
We omit the subscript $\cL$ when the considered LLL is unambiguous. 
The \lovasz Local Lemma \cite{LLL73}  states that $\Pr(\cap_{\mathcal{E}\in \cB}\bar{\mathcal{E}})>0$ holds if $epd<1$, or in other words, there exists an assignment to the variables that avoids all bad events.

 In the \emph{constructive \lovasz local lemma} one aims to compute such an \emph{feasible} assignment, avoiding all bad events.
 This is often under much stronger conditions on the relation of $p$ and $d$. The relation of $p$ and $d$ is referred to as the \emph{LLL criterion}.

\subparagraph{Constructive Distributed \lovasz Local Lemma}
In the distributed setting, the LLL instance $\cL$ is mapped to a communication network $G=(V,E)$.
We are given a function $\ell: \cB\cup \cV\rightarrow V$ that assigns each variable and each bad event to a node of the communication network.  We assume that for each variable $x\in \cV$, the node $\ell(x)$ knows the distribution of $x$, including the range $\mathsf{range}(x)$ of the variable. We also say that node $\ell(x)$ \emph{simulates} the variable/event $x$. For  a vertex $v\in V$, we call $l(v)=|\ell^{-1}(v)|$ the \emph{load} of vertex $v$. The \emph{(maximum) vertex load} of an LLL instance is $l=\max_{v\in V}l(v)$.

In the constructive distributed LLL, we execute a \LOCAL or \CONGEST algorithm on $G$ to compute a feasible assignment $\phi$. Afterwards, for each variable $x \in \cV$,  node $\ell(x)$ has to output $\phi(x)$. 

In general, the graph $G$ and the dependency graph $H_{\cL}$ do not have to coincide. 
However, distances between events in $H_{\cL}$ and the corresponding nodes in $G$ are in close relation, as formalized by the next definition. 
\begin{definition}[Locality]
A triple $(\cL,G,\ell)$ has \emph{locality} $\nu$ if $\dist_G(\ell(\cE),\ell(x))\le \nu$ for all events $\cE$ of $\cL$ and variables $x\in\vbl(\cE)$. 
\end{definition}

\myparagraph{(Partial) Assignments.}
 We use the value $\bot$ for variables that have not been set yet. A \emph{partial assignment $\phi$} of a set of variables $\cV$ is a function with domain $\cV$ satisfying $\phi(x)\in \mathsf{range}(x)\cup \{\bot\}$ for all $x\in \cV$. 
A partial assignment $\psi$ \emph{agrees} with another (partial) assignment $\phi$ if $\psi(x)=\phi(x)$ for all $x\notin \psi^{-1}(\bot)$, i.e., if all proper values assigned by $\psi$ match those of $\phi$.
 A \emph{retraction} $\psi$ of a partial assignment $\phi$ is a partial assignment that agrees with $\phi$.
For an event $\cE$ and a partial assignment $\phi$, we use the notation $\Pr(\cE \mid \phi)$ to mean that 
the probability is over assignments with which $\phi$ agrees;
in other words, the randomness is only over the variables in $\phi^{-1}(\bot)$.

\subparagraph{Simulatable Distributed \lovasz Local Lemma (CONGEST)}
\label{sec:simulatability}

\begin{definition}[Simulatability]
\label{def:simulatability}
We say an LLL $(\cL,G,\ell)$ is \emph{simulatable} in $\CONGEST$ if each of the following can be done in $\poly\log\log n$ rounds:
\begin{compactenum}
    \item \textbf{Test:} Test in parallel which events of $\cL$ hold (without preprocessing).
    \item \textbf{Min-aggregation:} Given $1$ bit in each event (variable), 
    each variable (event) can simultaneously find the minimum of the bits of its events (variables). 

    \medskip
    
    \item [] For the following items, it is sufficient if they hold in the setting that events and variables are given $O(\log\log n)$-bit identifiers\footnote{In general, for the whole LLL instance and for non-constant distances such identifiers do not exist, but our LLL algorithms only use the primitives in settings where they do exists and are available.} (that are unique within distance $4\nu$ in $G$):
    \item \textbf{Evaluate:} Given a partial assignment $\phi$, and partial assignments $\psi_1, \ldots, \psi_t$, $t = O(\log n)$, in which each variable knows its values (or $\bot$), each event $\cE$ of $\cL$ can simultaneously decide if
         $\Pr(\cE\mid \psi_i)\le \alpha \Pr(\cE\mid \phi)$
     holds, where $\alpha$ is a parameter known by all nodes of $G$. 
     \item \textbf{Min-aggregation:}  We can compute the following for $O(\log n)$ different instances in parallel: Given an $O(\log\log n)$-bit string in each event (variable), each variable (event) can simultaneously find the minimum of the strings for its events (for its variables). 
\end{compactenum}
\end{definition}

\subparagraph{Disjoint Variable Set LLLs}
\label{sec:tecOverviewtwoVariableLLL}
In a \emph{disjoint variable set LLLs} there are two disjoint sets of variables $\cV_1, \cV_2$ available for each event. In fact, we consider events $\cE$ that can be written as the conjunction of two events $\cE_1,\cE_2$ where $\vbl(\cE_i)=\cV_i$ and $\Pr(\cE_i)\le p$ holds for $i=1,2$. Note, that to avoid $\cE$ it is sufficient to avoid either $\cE_1$ or $\cE_2$.

\begin{restatable}{theorem}{thmTwoVariableSet}
\label{thm:LLLTwoSets}
There is a randomized \CONGEST algorithm that in $\poly\log\log n$ rounds w.h.p.\ solve any \underline{disjoint variable set LLL} of constant locality $\nu$ with dependency degree $d\le \poly\log n$ and bad event upper bound $p$. The  algorithm requires $p<d^{-(2+c_l)-(4c+12c_{\Delta}\nu)\log\log n}$, $l\le d^{c_l}$, $\Delta\le \log^c n$ for constants $c_l, c_{\Delta}\ge 1$, and that the LLL is simulatable. 
\end{restatable}

\subparagraph{Sampling LLLs}
\label{sec:tecOverviewsampleLLL}
In a \emph{binary LLL} the range of the variables is $\{\mathsf{black},\mathsf{white}\}$. We view the variables/nodes with black value as \emph{sampled}. Thus, we also refer to them as \emph{sampling LLLs}. 
The \emph{risk} of a bad event $\cE$ upper bounds the probability of a bad event to hold under a certain type of retractions of an assignment that avoided an associated event $\cE'$. 

\begin{restatable}[\risk]{definition}{defPromiseretractionCost}
\label{def:promiseretractionCost}
We say that an (associated) event $\cE'$ \emph{testifies}  \risk $x$ for some event $\cE\subseteq \cE'$ if 
\begin{align}
    \max\big\{\Pr(\cE'),\max_{\psi\in \mathsf{Respect}(\cE')}\{\Pr(\cE\mid \psi)\}\big\}\le x~.
\end{align}
The \emph{risk} of an event $\cE$ is the smallest risk testified by some event $\cE'\supseteq \cE$~.
\end{restatable}
Here, $\mathsf{Respect}(\cE')$ is the set of retractions of assignments avoiding $\cE'$, where either (i) no \black variables of $\cE'$ or (ii) all \white variables of $\cE'$ are retracted.  In our algorithms we will use several LLLs that have a low risk and hence can be solved with the following theorem.

\begin{restatable}{theorem}{thmpromiseLLL}
\label{thm:promiseLLL}
There is a randomized  \CONGEST algorithm that in $\poly\log\log n$ rounds w.h.p.\ solve any  LLL of constant locality $\nu$ with dependency degree $d\le \poly\log n$ and \underline{risk} $p$. The algorithm requires $p<d^{-(4+c_l)-(4c+12c\nu)\log\log n}$, $l\le d^{c_l}$, $\Delta\le \log^{c_{\Delta}} n$ for constants $c_l,c_{\Delta}\ge 1$ and that the LLL is simulatable.
\end{restatable}
Events \emph{favor} \black, or are monotone \emph{increasing}, if changing any value to \black does not decrease the conditional probability of the event, respectively. A typical example of a monotone increasing event is sampling a subset of vertices containing many non-edges in the neighborhood of each node. We use this problem in our procedure to color sparse nodes. 
A key point is that it is easy to bound the risk of monotone increasing events as shown in the following lemma from \cite{HMP24}.

\begin{restatable}{lemma}{lemRISKmonotone}
The \risk of a monotone increasing event $\cE$ is  $\Pr(\cE)$ testified by $\assoc(\cE)=\cE$. 
\label{L:monotone-incr}
\end{restatable}

Last but not least we will sample subsets of nodes satisfying certain degree bounds. The following lemma is helpful to bound the risk of such sampling LLLs. 

\begin{lemma}\label{L:incr-LLL} 
Consider a random variable $X$ that is a sum of independent binary random variables. 
For some threshold parameter $x>0$, let $\cE_x$ be the event that $X>x$ holds. 
Then, the \risk of $\cE_x$ is at most $\Pr(\cE_{\nicefrac{x}{2}})$ testified by $\cE_{\nicefrac{x}{2}}$.
\end{lemma}

\subsection{Generating Unit Slack for Sparse and Ordinary Nodes}
\label{sec:slackSparseSmallOrdinary}
The next lemma shows that the nodes in small ordinary cliques are somewhat sparse.  
As each large AC is a proper clique consisting of nodes with degree $\Delta$, we obtain the following. 
\begin{observation}[Small ordinary cliques are sparse]
\label{lem:smallOrdinarySparse}
Any node $v$ in an ordinary AC $C$ has at least $e_C\cdot (\Delta-3e_C)$ non-edges in its neighborhood. In particular, any small node has at least $\Delta^2/(2q(n))$ non-edges in its neighborhood. 
\end{observation}
\begin{proof}
Since $C$ is not easy, each of its $\Delta+1-e_C$ nodes have $e_C$ external neighbors.
Since $C$ is not difficult, it has no intrusive neighbor. Thus, each external neighbor of $v \in C$ has at most $2e_C$ neighbors in $C$, so at least $|C| - 2e_C \ge \Delta - 3e_C$ non-neighbors. Hence, the first claim.
A small node has $e_C \ge \Delta/q(n)$, implying the second claim.
\end{proof}

The task of this subsection is to prove the following lemma.

\lemSlackSparseSmallOrdinary*

\begin{proof}
Let $U=V^*\cup \mathcal{O}$ and $U'=\{v\in V^*\cup \mathcal{O}_s\mid N(v)\subseteq U\}\subseteq U$. Note that any node in 
$V^*\cup \mathcal{O}_s$ with a neighbor $w\notin U$ automatically has unit-slack in $G[U]$ as its neighbor $w$ is stalled to be colored later.
Thus we can concentrate on the vertices in $U'$. 

Each node $v\in U'\cap V^*$ is sparse and so
has $\eps^2\Delta^2$ non-edges in its induced neighborhood, which is within $G[U]$.  Each node in $U'\cap \mathcal{O}_s$ has at least $\Delta^2/(2 q(n))$ non-edges in its neighborhood in $G[U]$ by \Cref{lem:smallOrdinarySparse}. 
Let $\mu= c\log^4\log n\cdot \log \Delta$. 

We first use \Cref{thm:promiseLLL} (twice) to compute two sets $S_i\subseteq U$, $i=1,2$ satisfying the following properties:
\begin{enumerate}
\item $|S_i\cap N(v)|\le \mu$, for all $v\in V^*\cup \mathcal{O}$ 
\item $|S_i\cap M_C|\le \mu$, for all ordinary ACs $C$~,
\item Number of non-edges in $G[S_i\cap N(v)]$ is $\Omega(\log^5\log n \cdot \log^2 \Delta)$, for each $v\in U'$.
\end{enumerate}

In order to construct $S_1$ consider the process that samples each node $U$ into $S_1$ with probability $p=c\Delta^{-1}\cdot \log^4\log n\cdot \log \Delta$ for a suitable constant $c$. For a suitable constants $c_1$ introduce the following bad events.
\begin{enumerate}
\item For all $v\in V^*\cup \mathcal{O}$, event $\cE_v$ holds if $|S_i\cap N(v)|\ge 4\mu $
\item For each ordinary AC $C$, event $\cE_C$ holds if $|S_i\cap M_C|\geq 4\mu$, 
\item For each $v\in U'$, event $\cE'_v$ holds if the number of non edges in $G[S_i\cap N(v)]$ is less than $c_1\cdot \log^5\log n \cdot \log^2 \Delta$.
\end{enumerate}
\begin{claim}
\label{claim:sparseSetSampling}
The sampling of $S_1\subseteq U$ with probability $p$ and the aforementioned bad events is a simulatable LLL with risk $\Delta^{-c/50 \cdot \log\log n}$. 
\end{claim}
\begin{proof}
We first bound the risk of the events and then reason about simulatability. 
\begin{itemize}
        \item Fix some $v\in V^*\cup\mathcal{O}$.
        The expected number of neighbors in $S$ is $d(v) \cdot p / \Delta \le \mu$. 
        Hence, $\Pr(\cE_{v}) \le \exp\left(-2\mu/3\right)$ by Chernoff.
        Additionally, define an associated event $\assoc(\cE_{v})$ as the event that at most $2\mu$ neighbors are sampled. We have $\Pr(\assoc(\cE_{v,d})) \le \exp\left(-\mu/3\right)$. This bounds the \risk of $\cE_{v}$ to be at most $\Pr(\assoc(\cE_{v}))$ by  \Cref{L:incr-LLL}.
        \item The proof for bounding the risk of the event $\cE_C$ for each ordinary clique is identical to the proof for $\cE_v$ by considering the sampling status of the matching $M_C$ instead of the neighborhood of the respective node. 
        \item Fix a node $v\in U'$ and  let $\alpha = \bar{m}(N(v)\cap U)/\Delta^2$ be the fraction of non-edges of node $v$ in its neighborhood induced by $U$. \Cref{lem:smallOrdinarySparse}  shows $\alpha\geq \min\{\eps^2,1/(2q^2(n))\}=1/(2q^2(n))$, regardless of whether $v\in V*$ or $v\in \mathcal{O}_s$.
        
        Now, fix the constant $c_1$ such that the event $\cE'_v$ is the event that the number of non-edges in $G[N(v) \cap S_1]$ is less than $\overline{m}_{thres} = \alpha \mu^2 / 2$. Let $f$ be a random variable for the number of non-edges in the graph induced by $X = N(v) \cap S_1$. Apply the non-edge hitting lemma \cref{lem:nonEdgeHittingV2}, with $|X| \le \Delta $ and $\overline{\mu} \ge \alpha\Delta^2$. The lemma shows that the expected number of non-edges is $\E[f] \ge p^2 \overline{m} \ge \alpha \mu^2$ and that $f$ is also well concentrated. We obtain 
        $
        \Pr(\cE'_v) = \Pr(f \le \E[f] / 2) 
        \le \exp \left( -\frac{p \overline{m}}{5 |X|} \right) 
        \le \exp \left( -\frac{\alpha \mu}{5} \right) 
        $. $\cE'_v$ is a monotone increasing event.
        Hence, its \risk is at most $\Pr(\cE'_v)$ by \Cref{L:monotone-incr}, where the associated event $\assoc(\cE'_v)$ is $\cE'_v$ itself.
\end{itemize}
In summary the risk is upper bounded by $\max\{\exp\left(-2\mu/3\right),\exp\left(-\alpha\mu/5\right)\}\leq \Delta^{-c/50 \cdot \log\log n}$.

The simulatability of the first two types of events ($\cE_v$ for $v\in U$ and $\cE_C$ for ordinary cliques $C$) is immediate as it only counts the number of immediate neighbors of nodes and cliques, respectively. Here, the leader node  $\ell(\cE_C)$ can gather full information about the number of nodes in $S\cap M_C$ in any partial assignment sampling $S$. 

The lengthy proof of the simulatability of the event $\cE'_v$ for $v\in U'$ is word by word identical to the proof of the simulatability of a similar type of event in \cite[Lemma 8.4, arxiv version]{HMP24}. The crucial point is part 3 of the simulatability definition (\Cref{def:simulatability}) where $O(\log n)$ evaluations of conditional probabilities need to be done in parallel in the setting where locally unique IDs are represented with $O(\log\log n)$ bits. These small IDs are sufficient for a preprocessing that is done simultaneously for all instances and in which $v$ learns the whole topology of $G[N(v)\cap U]$. Once the topology is available, the sampling status of nodes $S\cap N(v)$ will reveal the number of non-edges in $v$'s sampled neighborhood, showing simulatability. 
\end{proof}

Due to \Cref{claim:sparseSetSampling}, we can apply \Cref{thm:promiseLLL} to solve the LLL in \CONGEST and compute a set $S_1$ with the required properties in $\poly\log \log n$ rounds, w.h.p. We proceed analogously for $S_2$ but compute it as a subset of $U\setminus S_1$. The remaining steps are identical except that constant $c_3$ is replaced with a smaller constant as removing the set $S_1$ from $U$ may reduce the sparsity of the nodes in $U'$. Still, the reduction is limited to a constant factor for the following reason: removing at most $O(\log^4\log n \cdot \log\Delta)$ nodes from the neighborhood of each node reduces the number of non-edges in each neighborhood by at most $O(\Delta\log^4\log n \cdot \log\Delta)=O(\Delta\log^5\log n)$. Thus a node in $U'\cap \mathcal{O}_s$ still has $\Delta^2/(2q(n))-O(\Delta\log^5\log n)\geq \Delta^2/(4q(n))$ non-edges available, where we used that $n$ is large enough and $\Delta=\omega(q(n)\log^5\log n)$ holds.
For nodes in $U'\cap V^*$, removing the nodes in $S_1$ from $U$ also removes less than half of the initially available $\eps^2\Delta^2$ non-edges.

With the two sets $S_1$ and $S_2$, we apply \Cref{lem:slackgenAlg} with two disjoint color palettes of size $\chi=\lfloor\Delta/2\rfloor$. The number of non-edges $\overline{m}$ in $G[S_i\cap N(v)]$ satisfies $\overline{m}/\chi = \Omega(\log\Delta \cdot \log\log n)$ as required. As a result, a subset $S\subseteq S_1\cup S_2\subseteq V^*\cup \mathcal{O}$ is colored, such that all nodes in $V^* \cup \mathcal{O}_s$ get slack. 
The second property of this lemma, stating that the number of nodes colored in $N(v)$ of the respective nodes and in $M_C$, follows from the bound on number of neighbors in $N(v)\cap (S_1\cup S_2)$ and $M_C\cap (S_1\cup S_2)$. The runtime immediately follows from \Cref{thm:promiseLLL} and \Cref{lem:slackgenAlg}.
\end{proof}

\subsection{Computing the Set $Z=Z_1\cup Z_2$}
\label{sec:findingZ}

\lemUsefulEdges*

We compute the sets $Z_1$ and $Z_2$ by two consecutive LLLs $\cL_1$ and $\cL_2$. In the first LLL, we compute the set $Z$, which we split into the two sets $Z_1$ and $Z_2$ in the second LLL.  

\begin{definition}[First sampling LLL]
\label{def:LLLL1}
We define the following sampling LLL $\cL_1$.
Let $X=\{v\in \mathcal{O}_l: \text {$v$ is uncolored after Step~1}\}$.

\begin{compactitem}   
\item \textbf{Variables:} Sample each node of $X$ with probability $q=1/30$ into $Z$. Denote $Y=X\setminus Z$. 
\item \textbf{Bad Events:}
\begin{compactenum}   
    \item For each $v\in \mathcal{O}$, there is a bad event $\cE_v$ stating that $|Z\cap N(v)|> 3q\Delta$. 
    \item For each important AC $C$, define an event $\cE_C$ that holds if fewer than $q^2(1-q)^3\Delta/20$ edges of $M_C$ are useful. 
\end{compactenum}
\item \textbf{Associated Events}
\begin{compactenum}   
    \item $\assoc(\cE_v)$: For each $v\in \mathcal{O}$, the bad event $\assoc(\cE_v)$ holds if $|Z\cap N(v)|> 3q\Delta/2 = \Delta/20$~,
    \item $\assoc(\cE_C)$:  The event holds if there are fewer than $q(1-q)\Delta/10$ useful edges or if there are fewer than  $(1-q)^2\Delta/10$ \white edges in $C$.
\end{compactenum}
 \item \textbf{Event/variable assignment} $\ell$: Each variable and each event $\cE_v$, $\assoc(\cE_v)$ are simulated by the corresponding node. The events $\cE_C$ and $\assoc(\cE_C)$ are simulated by the node of $C$ with the largest ID. 
\end{compactitem}
\end{definition}

Note that $\assoc(\cE_C)$ is of different nature from $\cE_C$.

\begin{lemma}
\label{lem:assocEventDelta}
We have the following upper bounds for the probabilities of the respective events.
\begin{compactenum}
    \item For all $v\in \mathcal{O}$: $\Pr(\assoc(\cE_v))\le\exp(-\Omega(\Delta))$.
    \item For all important ACs $C$: $\Pr(\assoc(\cE_C))\le \exp(-\Omega(\Delta)) $. 
\end{compactenum}
\end{lemma}
\begin{proof}
Throughout the proof we use that $q$ and $1-q$ are constant.

\textbf{Bounding $\Pr(\assoc(\cE_v))$:} As each node joins $Z$ independently with probability $q$, we have $E[|Z\cap N(v)|]\le q\Delta$, and the first bound follows from a Chernoff bound. 
   
\textbf{Bounding $\Pr(\assoc(\cE_C))$:}
Let $N_C$ be the arcs of $M_C$ that have both endpoints in $\mathcal{O}_l$ and uncolored after Step~1.  All heads of arcs in $M_C$ are already in $\mathcal{O}_l$, and by the definition of an important AC, at least $\Delta/12$ arcs in $M_C$ have their tails in $\mathcal{O}_l$. At most $O(\log^5\log n)$ of $V(M_C)$ are already colored. Thus,
$N_C$ contains at least $\Delta/12-O(\log^5\log n)\ge \Delta/20$ nodes. 
  
   Now, observe that the probability for an edge of $N_C$ to be useful is $q(1-q)$ and the expected number of useful edges is $q(1-q)|N_C|=q(1-q)\Delta/15$. This property is independent for different edges in $N_C$, so the claim regarding the number of useful edges follows from a Chernoff bound. Similarly, the probability for an edge to be \white is $(1-q)^2$, and the expected number of \white edges in $M_C$ is $(1-q)^2|N_C|=(1-q)^2\Delta/15$. The claim regarding \white edges then follows with a Chernoff bound.   
\end{proof}

\begin{lemma}
\label{lem:l1promiseRetractionCost}
$\cL_1$ is a sampling LLL with risk $\exp(-\Omega(\Delta))$ and dependency degree $O(\Delta^2)$. 
\end{lemma}
\begin{proof}
The probabilities of the associated events $\assoc(\cE_v)$
and $\assoc(\cE_C)$ 
are at most $\exp(-\Omega(\Delta))$ by \Cref{lem:assocEventDelta}.

The dependency degree can be bounded as follows. Each variable of a node stating whether the node is \white or \black only appears in the events $\cE_v$ and $\assoc(\cE_v)$ of adjacent nodes and in the events $\cE_C$ and $\assoc(\cE_C)$ of adjacent ACs, bounding the variable degree by $O(\Delta)$. We have $|\vbl(\cE_V)|\le \Delta$ and each event $\cE_C$ depends on two variables for each edge in $M_C$. As $|M_C|\le \Delta$, we obtain that each event depends on at most $O(\Delta)$ variables and the dependency degree can be upper bounded by $O(\Delta^2)$.

Via \Cref{L:incr-LLL} we obtain that $\assoc(\cE_v)$ testifies that $\cE_v$ has risk $\exp(-\Omega(\Delta))$. 

Next, we fix an important AC $C$ and reason that $\assoc(\cE_C)$ testifies that $\cE_C$ has risk $\exp(-\Omega(\Delta))$.   First note that $\cE_C\subseteq \assoc(\cE_C)$, as required by \Cref{def:promiseretractionCost}. Let $\psi\in \mathsf{Respect}(\assoc(\cE_C))$, namely $\psi$ is a retraction of an assignment $\phi$ under which $\assoc(\cE_C)$ is avoided. By the definition of $\mathsf{Respect}(\assoc(\cE_C))$, the set of retracted variables is in one of the following two cases: 1) The set of retracted variables contains no variables of $\vbl(\cE_C)$ that were \black under $\phi$, or 2) The set of retracted variables contains all variables of $\vbl(\cE_C)$ that were \white under $\phi$.

Let us first consider the second case. As $\assoc(\cE_C)$ is avoided under $\phi$, under the assignment $\phi$ at least $(1-q)^2\Delta/10$  edges of $M_C$ are white. In the second case, all of these obtain fresh randomness, and each of them is useful independently with probability $q(1-q)$. Thus, in expectation, at least $q(1-q)^3\Delta/10$ of them are useful. With a Chernoff bound, we obtain that the probability of $\cE_C$ to happen in the second case is at most $\exp(-\Omega(\Delta))$.

Now consider the first case.  As $\assoc(\cE_C)$ is avoided under $\phi$, under the assignment $\phi$ at least $q(1-q)\Delta/10$  edges of $M_C$ are useful. Let $U\subseteq C$ be the set of nodes in those useful edges that are contained in $C$. Note that all nodes in $U$ are \white under $\phi$. Let $U_1\subseteq U$ be the nodes that are also \white under $\psi$ and let $U_2\subseteq U$ be the nodes that evaluate to $\bot$ under $\psi$, i.e., got retracted. Nodes in $U_2$ are \black/\white with probability $q$ and $1-q$, respectively. Let $U^w_2$ be the random variable describing the number of nodes of $U_2$ set to \white in this process.
Let $\alpha$ be the random variable describing the number of useful edges in $M_C$ after that process. We obtain $\E[\alpha]\ge \E[|U_1|+|U^w_2|]=|U_1|+(1-q)|U_2|\ge |U_1|+|U_2|/2\ge |U|/2\ge q(1-q)\Delta/10$~, where we used that $(1-q)\ge 1/2$. 
The event $\cE_C$ holds if  $\alpha\le q^2(1-q)^3\Delta/20$, which is smaller than $\E[\alpha]/2$. Hence, we obtain that $\cE_C$ happens with probability at most  $\exp(-\Omega(\Delta))$ by a Chernoff bound. 
\end{proof}

\begin{lemma}
\label{lem:l1simulatable}
    $\cL_1$ is simulatable. 
\end{lemma}
\begin{proof}
Each event $\cE_v$ depends only on variables that are immediately incident to the node $\cE_v:\vbl(\cE_v)\mapsto \{\true,\false\}$ is a function counting the number of nodes that is known to $\ell(\cE_v)$. Hence, the simulatability condition holds for $\cE_v$. For $\cE_C$ all variables are simulated by nodes that are immediately incident to the AC $C$ and full knowledge about these variables can be relayed to the leader in the AC that simulates event $\cE_C$. Again, whether the event $\cE_C$ holds can be evaluated with the values of the variables and the edges in $M_C$, also for all conditional probabilities of partial assignments, as $\ell(\cE_C)$ has full knowledge of the function $\cE_C:\vbl(\cE_C)\mapsto \{\true,\false\}$.     
\end{proof}

Let $x=q^2(1-q)^3\Delta/20$ be the threshold of the number of useful edges that are guaranteed in each $M_C$ for each important AC by $\cL_1$ (see \Cref{def:LLLL1}). The second LLL is significantly simpler and given by the following definition. 
\begin{definition}[Second sampling LLL]
We define the following LLL $\cL_2$. We split $Z$ into two sets $Z_1$ and $Z_2$ where each node in $Z$ flips an unbiased coin which set to join.
There are bad events $\cE_{C,i}$, $i=1,2$ and for each important AC $C$, that hold if there are fewer than $x/3$ useful edges in $U(C,Z_i)$, respectively. 
\end{definition}

This LLL can be solved by a result in \cite{HMP24}.

\begin{lemma}
\label{lem:l2}
There is a $\poly\log\log n$-round \CONGEST algorithm for $\cL_2$.
\end{lemma}

\begin{proof}
Form the bipartite graph $H=(U,Z,E_H)$ with the nodes of $Z$ on one side and a node $u_C$ for each important AC $C$ on the other side.
There is an edge $(u_C, z)$ for each useful arc $\overrightarrow{vz} \in U(C,Z)$.
Each node $u_C$ has degree at least $x \ge q^2 \Delta/30 = \Omega(\Delta)$ (by Lemma 4.6), while each node $z \in Z$ has degree at most $\Delta/q(n)$ (as $z$ is large). 
Splitting the subset $Z$ into two parts such that each node $v \in U$ has between $d(v)/2$ and $3d(v)/2$ neighbors into each part is a vertex subset-splitting problem formulated as bounded-risk LLL and solved in Lemma D.11(1) of \cite{HMP24}.

We only need to verify that this problem remains simulatable in our embedded setting.
 The problem is simulatable because the node $\ell(\cE_C)$ can obtain full knowledge of any partial assignment of $\vbl(\cE_C)$ and knows the function $\cE_C:\vbl(\cE_C)\mapsto \{\true,\false\}$. 
\end{proof}

\begin{proof}[Proof of \Cref{lem:usefulEdges}]
    First, apply \Cref{thm:promiseLLL} in order to solve $\cL_1$ in $\poly\log\log n$ rounds yielding a set $Z$ that avoids all bad events of $\cL_1$. The conditions of the theorem are met by \Cref{lem:l1simulatable,lem:l1promiseRetractionCost} and as $\Delta\ge \log^{10}\log n$ implies that the criterion is strong enough. We split the set $Z$ into $Z_1$ and $Z_2$ by solving $\cL_2$, by 
    \cref{lem:l2}
    The requirements of the theorem are satisfied 
    as $\Delta\ge \log^{10}\log n$. 

    The degree bound immediately follows from the conditions on $Z$ imposed in the neighborhood of each vertex $v$ by $\cL_1$ (note that $Z=Z_1\cup Z_2)$. The second property follows from the avoided events of $\cL_2$ for each important AC. 
\end{proof}

\subsection{Forming Triples}
\label{sec:TripleForming}

\lemTripleLemma*

We model the problem of finding $z_C$ for each important AC $C$ as a disjoint variable set LLL, where the disjoint sets $Z_1$ and $Z_2$ give rise to two disjoint sets of variables. Note that the respective nodes $x_C$ and $y_C$ will only be computed in the sequel via a deterministic method. Recall, that $Z=Z_1\cup Z_2$. 
\begin{definition}
\label{def:LLLL3}
    Define the following disjoint variable set LLL $\cL_3$. 

\begin{itemize}
    \item \textbf{Variables}: 
     For each important AC $C$ and each useful arc $\overrightarrow{vz} \in U(C,Z)$,
     there is a binary random variable $x_{vz}$ that assumes $1$ with probability $p_3 = q(n)/\Delta$.
    \item \textbf{Events:} We call a useful arc $\overrightarrow{vz} \in U(C,Z)$
    \emph{successful} 
        if AC $C$ activated $\overrightarrow{vz}$ and no other AC activated an edge $\overrightarrow{wz}$ (for some $w$).
        There is one bad event $\cE_C$ for each important AC $C$ that holds if there is no successful edge for $C$. We introduce corresponding events $\cE_{C,1}$ and $\cE_{C,2}$ restricted to arcs in $U(C,Z_i)$, respectively. We have  $\cE=\cE_{C,1}\cap \cE_{C,2}$. 
    \item The home node of $\cE_C$ is $\ell(\cE_C)=v_C$ where $v_C$ is the node of $C$ with largest ID. The home node of $x_{vz}$ is $\ell(x_{vz})=z$.
\end{itemize}
\end{definition}

\begin{lemma}
For each important AC $C$ and each $i=1,2$, we have $\Pr(\cE_{C,i})\le 2^{-\Omega(q(n))}$ and the dependency degree of $\cL_3$ is upper bounded by $d=O(\Delta^3)$.
\end{lemma}
\begin{proof}
Fix an important AC $C$ and $i \in {1,2}$.
For important arc $\overrightarrow{vz} \in U(C,Z_i)$, let $A_{vz}$ be the event that 
$\overrightarrow{vz}$ is successful. Observe that $A_{vz}$ depends only on arcs with $z$ as tail. It holds if $\overrightarrow{vz}$ is activated (i.e., has $x_{vz} = 1$) while the other arcs with $z$ as tail are not activated.
The external degree of each $z\in Z\subseteq \mathcal{O}_l$ is at most $\Delta/q(n)$, as it is large, so at most that many useful arcs have $z$ as tail.
Thus, $\Pr(A_{vz}) \ge \Pr(x_{vz} = 1) \cdot (1-p_3)^{\Delta/10} = p_3 (1-q(n)/\Delta)^{\Delta/q(n)} \ge p_3 (1/4)^{1/10} \ge 0.9 p_3$.

The events $A_{vz}$ and $A_{v'z'}$ are independent, as each they involve disjoint sets of arcs. 
The bad event $\cE_{C,i}$ holds only when no useful edge in $U(C,Z_i)$ becomes successful, which occurs with probability 
\begin{align*}
    \Pr(\cE_{C,i}) & = \prod_{\overrightarrow{vz} \in U(C,Z_i)} \Pr(\overline{A_{vz}}) 
   \le (1-0.9p_3)^{|U(C,Z_i)|} \\
   &\le (1-0.9q(n)/\Delta)^{q^2 (1-q)^3 \Delta/60} \le e^{-\Omega(q(n))} 
\end{align*}

The dependency degree is upper bounded by $O(\Delta^3)$ because the events of an AC only share variables with ACs that are within distance $2$ from one of the $\Delta$ nodes of the AC. 
\end{proof}

\begin{lemma}
$\cL_3$ is simulatable.
\end{lemma}
\begin{proof}
The respective nodes can check in $O(1)$ rounds whether their events hold by an assignment and the aggregation primitives can be implemented efficiently as all variables are in distance at most $2$ from the ACs.

We next reason why we can compute the conditional probabilities of \Cref{def:simulatability}. Let $\psi$ be any partial assignment. To compute the conditional probabilities $\Pr(\cE_C\mid \psi)$, $\Pr(\cE_{C,1}\mid \psi)$ and $\Pr(\cE_{C,2}\mid \psi)$, the node holding the respective event needs to compute the probability that one of the useful edges becomes successful for $C$. The conditional probability is $0$ if there already is a useful edge that is successful for $C$. For all other useful edges in $M_C$, the probability of becoming successful is independent as the activation by $C$ happens independently, and also, all activations from other ACs do influence at most one useful edge in $M_C$. The probability to become successful for a single useful edge $(v,z)$, $v\in C, z\in Z$ conditioned on $\psi$ can be computed from knowing whether $C$ activated $(v,z)$ in $\psi$, whether any other edge $(v',z)$ is activated in $\psi$ and from the number of useful edges of other ACs with endpoint $z$ that evaluate to $\bot$ under $\psi$. The nodes $\ell(\cE_C)=\ell(\cE_{C,1})=\ell(\cE_{C,1})$ can learn all this information in $O(1)$ rounds using $O(\log\log n)$ bits of communication per edge (here we use that $\Delta\le\poly\log n$ to communicate the aforementioned number efficiently). This can be performed in parallel for the events of all important ACs. Knowing the probability for each edge to become successful, the respective node can compute the conditional probability for the event. This proof also subsumes that the events can be evaluated efficiently, as we did not require the locally unique IDs from a smaller ID space that are given by \Cref{def:simulatability}. 
\end{proof}

\begin{proof}[Proof of \Cref{lem:tripleLemma}]
Fix an important AC $C$. First, we use \Cref{thm:LLLTwoSets} to solve $\cL_3$ with the sets $Z_1$ and $Z_2$ given from \Cref{lem:usefulEdges}. The algorithm runs in $\poly\log\log n$ rounds and works w.h.p. It provides us with a successful edge  $(y_C,z_C)$ for the AC $C$ (see \Cref{def:LLLL3}). Next, we show that we can deterministically compute a node $x_C$ to form the triple of nodes as required for \Cref{lem:tripleLemma}  in $O(1)$ rounds. 

The nodes in $C$ that cannot be used for $x_C$ are those that are either: a) neighbors of $z_C$, b) already colored, or c) function as $z_{C'}$ for another important AC $C$.
As $z_C$ is large, it 
has at most $\Delta/q(n)$ neighbors in $C$.
By \cref{lem:slackSparseSmallOrdinary}, at most $O(\log^4\log n \cdot \log\Delta) = O(\log^5\log n)$ nodes of $C$ are already colored.

By \Cref{lem:usefulEdges}, we obtain at most $|Z\cap C|\le |N(v) \cap Z| \le \Delta/10$ nodes in $C$ are candidates for being the outside node in a triple.
Hence, at least $|C| - \Delta/q(n) - O(\log^5\log n) - \Delta/10 \ge \Delta/2$ nodes in $C$ will do as a $x_C$-node.
   \item The graph induced by $\{z_c : C \text{ is an important AC}\}\subseteq Z\subseteq \mathcal{O}_l$ has maximum degree $\Delta/10$ as \Cref{lem:usefulEdges} ensures that $|N(v)\cap Z|\le \Delta/10$ for all $v\in \mathcal{O}_l$. 
\end{proof}

\bibliographystyle{alpha}
\bibliography{references-arxiv}

\newcommand{\etalchar}[1]{$^{#1}$}
\begin{thebibliography}{GHKM18}

\bibitem[AKM22]{AKM22}
Sepehr Assadi, Pankaj Kumar, and Parth Mittal.
\newblock Brooks’ theorem in graph streams: a single-pass semi-streaming
  algorithm for {$\Delta$}-coloring.
\newblock In {\em Proceedings of the 54th Annual ACM SIGACT Symposium on Theory
  of Computing}, pages 234--247, 2022.

\bibitem[Bar15]{barenboim15}
L.~Barenboim.
\newblock Deterministic ({$\Delta$} + 1)-coloring in sublinear (in {$\Delta$})
  time in static, dynamic and faulty networks.
\newblock In {\em Proc.~34th ACM Symposium on Principles of Distributed
  Computing (PODC)}, pages 345--354, 2015.

\bibitem[BBKO22]{BBKO2021hideandseek}
Alkida Balliu, Sebastian Brandt, Fabian Kuhn, and Dennis Olivetti.
\newblock Distributed ${\Delta}$-coloring plays hide-and-seek.
\newblock In {\em Proc.\ 54th {ACM} Symp.\ on Theory of Computing (STOC)},
  2022.

\bibitem[BCM{\etalchar{+}}21]{BCMOS21}
Alkida Balliu, Keren Censor{-}Hillel, Yannic Maus, Dennis Olivetti, and Jukka
  Suomela.
\newblock Locally checkable labelings with small messages.
\newblock In Seth Gilbert, editor, {\em 35th International Symposium on
  Distributed Computing, {DISC} 2021, October 4-8, 2021, Freiburg, Germany
  (Virtual Conference)}, volume 209 of {\em LIPIcs}, pages 8:1--8:18. Schloss
  Dagstuhl - Leibniz-Zentrum f{\"{u}}r Informatik, 2021.

\bibitem[BE13]{barenboimelkin_book}
Leonid Barenboim and Michael Elkin.
\newblock {\em Distributed Graph Coloring: Fundamentals and Recent
  Developments}.
\newblock Morgan \& Claypool Publishers, 2013.

\bibitem[BE19]{BamasEsperet19}
{\'{E}}tienne Bamas and Louis Esperet.
\newblock Distributed coloring of graphs with an optimal number of colors.
\newblock volume 126 of {\em LIPIcs}, pages 10:1--10:15. {LZI}, 2019.

\bibitem[BEPS16]{BEPSv3}
Leonid Barenboim, Michael Elkin, Seth Pettie, and Johannes Schneider.
\newblock The locality of distributed symmetry breaking.
\newblock {\em Journal of the ACM}, 63(3):20:1--20:45, 2016.

\bibitem[BFH{\etalchar{+}}16]{brandt2016LLL}
Sebastian Brandt, Orr Fischer, Juho Hirvonen, Barbara Keller, Tuomo
  Lempi{\"a}inen, Joel Rybicki, Jukka Suomela, and Jara Uitto.
\newblock A lower bound for the distributed {Lov{\'a}sz} local lemma.
\newblock In {\em Proc. 48th ACM Symposium on Theory of Computing (STOC 2016)},
  pages 479--488. ACM, 2016.

\bibitem[BKM20]{BKM20}
Philipp Bamberger, Fabian Kuhn, and Yannic Maus.
\newblock Efficient deterministic distributed coloring with small bandwidth.
\newblock In {\em {PODC} '20: {ACM} Symposium on Principles of Distributed
  Computing, Virtual Event, Italy, August 3-7, 2020}, pages 243--252, 2020.

\bibitem[Bro41]{brooks_1941}
R.~Leonard Brooks.
\newblock On colouring the nodes of a network.
\newblock {\em Mathematical Proceedings of the Cambridge Philosophical
  Society}, 37(2):194–197, 1941.

\bibitem[CCDM24]{CCDM24}
Sam Coy, Artur Czumaj, Peter Davies, and Gopinath Mishra.
\newblock Parallel derandomization for coloring, 2024.
\newblock Note: \url{https://arxiv.org/abs/2302.04378v1} contains the
  Delta-coloring algorithm.

\bibitem[CHL{\etalchar{+}}20]{CHLPU20}
Yi{-}Jun Chang, Qizheng He, Wenzheng Li, Seth Pettie, and Jara Uitto.
\newblock Distributed edge coloring and a special case of the constructive
  {L}ov{\'{a}}sz local lemma.
\newblock {\em {ACM} Trans. Algorithms}, 2020.

\bibitem[CLP18]{CLP18}
Yi{-}Jun Chang, Wenzheng Li, and Seth Pettie.
\newblock An optimal distributed ({\(\Delta\)}+1)-coloring algorithm?
\newblock In {\em Proceedings of the ACM Symposium on Theory of Computing
  (STOC)}, pages 445--456, 2018.

\bibitem[CM19]{CM19}
Shiri Chechik and Doron Mukhtar.
\newblock Optimal distributed coloring algorithms for planar graphs in the
  {LOCAL} model.
\newblock In Timothy~M. Chan, editor, {\em Proceedings of the Thirtieth Annual
  {ACM-SIAM} Symposium on Discrete Algorithms, {SODA} 2019, San Diego,
  California, USA, January 6-9, 2019}, pages 787--804. {SIAM}, 2019.

\bibitem[CP19]{CP19}
Yi{-}Jun Chang and Seth Pettie.
\newblock A time hierarchy theorem for the {LOCAL} model.
\newblock {\em SIAM J. Comput.}, 48(1):33--69, 2019.

\bibitem[CPS17]{CPS17}
Kai{-}Min Chung, Seth Pettie, and Hsin{-}Hao Su.
\newblock Distributed algorithms for the {Lov{\'{a}}sz} local lemma and graph
  coloring.
\newblock {\em Distributed Comput.}, 30(4):261--280, 2017.

\bibitem[Dav23]{Davies23}
Peter Davies.
\newblock Improved distributed algorithms for the {Lov{\'a}sz} local lemma and
  edge coloring.
\newblock In {\em Proceedings of the 2023 Annual ACM-SIAM Symposium on Discrete
  Algorithms (SODA)}, pages 4273--4295. SIAM, 2023.

\bibitem[EL74]{LLL73}
Paul Erd{\"o}s and L\'{a}szl{\'o} Lov{\'a}sz.
\newblock {Problems and Results on 3-chromatic Hypergraphs and some Related
  Questions}.
\newblock {\em Colloquia Mathematica Societatis J{\'a}nos Bolyai}, pages
  609--627, 1974.

\bibitem[EPS15]{EPS15}
Michael Elkin, Seth Pettie, and Hsin{-}Hao Su.
\newblock (2{\(\Delta-1\)})-edge-coloring is much easier than maximal matching
  in the distributed setting.
\newblock In {\em Proceedings of the Twenty-Sixth Annual {ACM-SIAM} Symposium
  on Discrete Algorithms, {SODA} 2015, San Diego, CA, USA, January 4-6, 2015},
  pages 355--370, 2015.

\bibitem[FG17]{FGLLL17}
Manuela Fischer and Mohsen Ghaffari.
\newblock {Sublogarithmic Distributed Algorithms for Lov{\'a}sz Local Lemma,
  and the Complexity Hierarchy}.
\newblock In {\em the Proceedings of the 31st International Symposium on
  Distributed Computing (DISC)}, pages 18:1--18:16, 2017.

\bibitem[FHK16]{FHK}
Pierre Fraigniaud, Marc Heinrich, and Adrian Kosowski.
\newblock Local conflict coloring.
\newblock In {\em Proceedings of the IEEE Symposium on Foundations of Computer
  Science (FOCS)}, pages 625--634, 2016.

\bibitem[FHM23]{FHM23}
Manuela Fischer, Magnús~M. Halldórsson, and Yannic Maus.
\newblock Fast distributed {Brooks'} theorem.
\newblock In {\em Proceedings of the SIAM-ACM Symposium on Discrete Algorithms
  (SODA)}, pages 2567--2588, 2023.

\bibitem[Fis17]{Fischer17}
Manuela Fischer.
\newblock Improved deterministic distributed matching via rounding.
\newblock In Andr{\'{e}}a~W. Richa, editor, {\em 31st International Symposium
  on Distributed Computing, {DISC} 2017, October 16-20, 2017, Vienna, Austria},
  volume~91 of {\em LIPIcs}, pages 17:1--17:15. Schloss Dagstuhl -
  Leibniz-Zentrum f{\"{u}}r Informatik, 2017.

\bibitem[FK23]{FK23}
Marc Fuchs and Fabian Kuhn.
\newblock List defective colorings: Distributed algorithms and applications.
\newblock In Rotem Oshman, editor, {\em 37th International Symposium on
  Distributed Computing, {DISC} 2023, October 10-12, 2023, L'Aquila, Italy},
  volume 281 of {\em LIPIcs}, pages 22:1--22:23. Schloss Dagstuhl -
  Leibniz-Zentrum f{\"{u}}r Informatik, 2023.

\bibitem[Gha19]{ghaffari19_MIS}
Mohsen Ghaffari.
\newblock Distributed maximal independent set using small messages.
\newblock In {\em Proc.\ 30th Symp.\ on Discrete Algorithms (SODA)}, pages
  805--820, 2019.

\bibitem[GHK18]{GHK18}
Mohsen Ghaffari, David~G. Harris, and Fabian Kuhn.
\newblock On derandomizing local distributed algorithms.
\newblock In {\em 59th {IEEE} Annual Symposium on Foundations of Computer
  Science, {FOCS} 2018, Paris, France, October 7-9, 2018}, pages 662--673,
  2018.

\bibitem[GHKM18]{GHKM18}
Mohsen Ghaffari, Juho Hirvonen, Fabian Kuhn, and Yannic Maus.
\newblock Improved distributed delta-coloring.
\newblock In {\em Proceedings of the 2018 {ACM} Symposium on Principles of
  Distributed Computing, {PODC} 2018, Egham, United Kingdom, July 23-27, 2018},
  pages 427--436, 2018.

\bibitem[GK21]{GK20}
Mohsen Ghaffari and Fabian Kuhn.
\newblock Deterministic distributed vertex coloring: Simpler, faster, and
  without network decomposition.
\newblock In {\em Proceedings of the IEEE Symposium on Foundations of Computer
  Science (FOCS)}, pages 1009--1020, 2021.

\bibitem[HKMT21]{HKMT21}
Magn{\'{u}}s~M. Halld{\'{o}}rsson, Fabian Kuhn, Yannic Maus, and Tigran
  Tonoyan.
\newblock Efficient randomized distributed coloring in {CONGEST}.
\newblock In {\em Proceedings of the ACM Symposium on Theory of Computing
  (STOC)}, pages 1180--1193, 2021.
\newblock Full version at CoRR abs/2105.04700.

\bibitem[HKNT22]{HKNT22}
Magn{\'{u}}s~M. Halld{\'{o}}rsson, Fabian Kuhn, Alexandre Nolin, and Tigran
  Tonoyan.
\newblock Near-optimal distributed degree+1 coloring.
\newblock In Stefano Leonardi and Anupam Gupta, editors, {\em {STOC} '22: 54th
  Annual {ACM} {SIGACT} Symposium on Theory of Computing, Rome, Italy, June 20
  - 24, 2022}, pages 450--463. {ACM}, 2022.

\bibitem[HMN22]{HMN22}
Magn{\'{u}}s~M. Halld{\'{o}}rsson, Yannic Maus, and Alexandre Nolin.
\newblock Fast distributed vertex splitting with applications.
\newblock In Christian Scheideler, editor, {\em 36th International Symposium on
  Distributed Computing, {DISC} 2022, October 25-27, 2022, Augusta, Georgia,
  {USA}}, volume 246 of {\em LIPIcs}, pages 26:1--26:24. Schloss Dagstuhl -
  Leibniz-Zentrum f{\"{u}}r Informatik, 2022.

\bibitem[HMP24]{HMP24}
Magnús~M. Halldórsson, Yannic Maus, and Saku Peltonen.
\newblock Distributed {Lovász} local lemma under bandwidth limitations, 2024.

\bibitem[HN21]{HN21}
Magn{\'{u}}s~M. Halld{\'{o}}rsson and Alexandre Nolin.
\newblock Superfast coloring in {CONGEST} via efficient color sampling.
\newblock In Tomasz Jurdzinski and Stefan Schmid, editors, {\em Structural
  Information and Communication Complexity - 28th International Colloquium,
  {SIROCCO} 2021, Wroc{\l}aw, Poland, June 28 - July 1, 2021, Proceedings},
  volume 12810 of {\em Lecture Notes in Computer Science}, pages 68--83.
  Springer, 2021.

\bibitem[HNT22]{HNT22}
Magnús~M. Halldórsson, Alexandre Nolin, and Tigran Tonoyan.
\newblock Overcoming congestion in distributed coloring.
\newblock In {\em Proceedings of the ACM Symposium on Principles of Distributed
  Computing (PODC)}, pages 26--36. {ACM}, 2022.

\bibitem[HSS18]{HSS18}
David~G. Harris, Johannes Schneider, and Hsin-Hao Su.
\newblock {Distributed {($\Delta + 1$)}-coloring in sublogarithmic rounds}.
\newblock {\em Journal of the ACM}, 65:19:1--19:21, 2018.

\bibitem[Joh99]{johansson99}
{\"{O}}jvind Johansson.
\newblock Simple distributed {$\Delta+1$}-coloring of graphs.
\newblock {\em Inf. Process. Lett.}, 70(5):229--232, 1999.

\bibitem[Lin92]{linial92}
Nati Linial.
\newblock Locality in distributed graph algorithms.
\newblock {\em SIAM Journal on Computing}, 21(1):193--201, 1992.

\bibitem[MPU23]{MPU23}
Yannic Maus, Saku Peltonen, and Jara Uitto.
\newblock Distributed symmetry breaking on power graphs via sparsification.
\newblock In {\em Proceedings of the 2023 ACM Symposium on Principles of
  Distributed Computing}, PODC '23, page 157–167, New York, NY, USA, 2023.
  Association for Computing Machinery.

\bibitem[MT20]{MT20}
Yannic Maus and Tigran Tonoyan.
\newblock Local conflict coloring revisited: Linial for lists.
\newblock In Hagit Attiya, editor, {\em 34th International Symposium on
  Distributed Computing, {DISC} 2020, October 12-16, 2020, Virtual Conference},
  volume 179 of {\em LIPIcs}, pages 16:1--16:18. Schloss Dagstuhl -
  Leibniz-Zentrum f{\"{u}}r Informatik, 2020.

\bibitem[MU21]{MU21}
Yannic Maus and Jara Uitto.
\newblock Efficient {CONGEST} algorithms for the {Lov{\'{a}}sz} local lemma.
\newblock In Seth Gilbert, editor, {\em Proceedings of the International
  Symposium on Distributed Computing (DISC)}, volume 209 of {\em LIPIcs}, pages
  31:1--31:19. Schloss Dagstuhl - Leibniz-Zentrum f{\"{u}}r Informatik, 2021.

\bibitem[Pos19]{Postle19}
Luke Postle.
\newblock Linear-time and efficient distributed algorithms for list coloring
  graphs on surfaces.
\newblock In David Zuckerman, editor, {\em 60th {IEEE} Annual Symposium on
  Foundations of Computer Science, {FOCS} 2019, Baltimore, Maryland, USA,
  November 9-12, 2019}, pages 929--941. {IEEE} Computer Society, 2019.

\bibitem[PS95]{PS95}
Alessandro Panconesi and Aravind Srinivasan.
\newblock The local nature of {$\Delta$}-coloring and its algorithmic
  applications.
\newblock {\em Combinatorica}, 15(2):255--280, 1995.

\bibitem[RG20]{RG20}
V{\'{a}}clav Rozho\v{n} and Mohsen Ghaffari.
\newblock Polylogarithmic-time deterministic network decomposition and
  distributed derandomization.
\newblock In {\em Proceedings of the ACM Symposium on Theory of Computing
  (STOC)}, pages 350--363, 2020.

\end{thebibliography}

\appendix

\section{Concentration Bounds}

\begin{lemma}[Chernoff bounds]\label{lem:basicchernoff}
Let $\{X_i\}_{i=1}^r$ be a family of independent binary random variables with $\Pr[X_i=1]=q_i$, and let $X=\sum_{i=1}^r X_i$. For any $\delta>0$, 
\[
\Pr \left( |X-\Exp[X]|\ge \delta\Exp[X] \right) \le 2\exp(-\min(\delta,\delta^2) \Exp[X]/3) \ .
\]
\end{lemma}

\section{Further Supplementary Material from \cite{HMP24}}
\label{app:furtherMaterial}

\subparagraph{Slack generation with two given sets.}
The following lemma shows that one can compute a partial coloring of the nodes in two given sets $S_1$ and $S_2$ such that any node $v$ that has sufficiently many non-edges in $G[N(v)\cap S_i]$, $i=1,2$ obtain slack. We use it in \Cref{sec:slackSparseSmallOrdinary} and it is proven in \cite{HMP24}.

\begin{lemma}[\cite{HMP24}]
    \label{lem:slackgenAlg}
    Let $\Delta_s = O(\poly\log\log n)$. Let $\overline{m}$ and $\chi=O(\Delta)$ be positive integers such that $\overline{m} / \chi = \Omega(\log \Delta \cdot \log\log n)$ and $\chi \ge c'\Delta_s$ for some constant $c'$. Let $W \subseteq V$ and  let $S_1, S_2 \subset V$ be disjoint sets such that for $i=1,2$, 
    \begin{itemize}
        \item $\forall v \in (W \cup S_1 \cup S_2): d_{S_i}(v) \le \Delta_s$,
        \item $\forall v \in W$: the number of non-edges in $N(v) \cap S_i$ is at least $\overline{m}$
    \end{itemize}
    There is a randomized \CONGEST algorithm that w.h.p. colors a subset of $S_1 \cup S_2$ using a palette of size $2\chi$  such that every node in $W$ has at least $e^{-3/c'} \overline{m} /(50 \chi) = \Omega(\overline{m}/\chi)$ same-colored neighbors. Every node in $W, S_1, S_2$ has at most $2\Delta_s$ of its neighbors colored. 
\end{lemma}

\subparagraph{Non-edge hitting lemma.}
An expected $p^2$-fraction of the non-edges is preserved when sampling nodes into a set with probability $p$.
The following lemma shows that the probability of deviating from this expectation is small.
\begin{restatable}[Non-edge hitting lemma \cite{HMP24}]{lemma}{lemNonEdgeHitting}
	Let $G$ be a graph on the vertex set $X$ with $\overline{m}$ non-edges. Sample each node of $X$ with probability $p$ into a set $S$ and let $f$ be the random variable describing the number of non-edges in $G[S]$. 
	Then we have $\Pr(f \le p^2 \overline{m}/2) \le \exp \left(-p \overline{m} / 5 |X| \right)$.
\label{lem:nonEdgeHittingV2}
\end{restatable}
\end{document}